\documentclass[12pt, twoside]{IEEEtran} \onecolumn 

\pdfoutput=1

\usepackage{hyperref}

\title{A Correctness Result for Online Robust PCA}

\author{
Brian~Lois,~\IEEEmembership{Graduate~Student~Member,~IEEE}\\
Namrata~Vaswani,~\IEEEmembership{Senior Member,~IEEE}
\thanks{B. Lois is with the Mathematics and ECpE departments, and N. Vaswani is with the ECpE department at Iowa State University.
Email: \{blois,namrata\}@iastate.edu.}
\thanks{This work was partly supported by NSF grant CCF-1117125. 
}
}

\usepackage{amsmath, amssymb, amsthm, algorithm, algorithmic,graphicx}
\usepackage{url,subfigure}
\usepackage{bm}

\usepackage{enumerate}

\usepackage{color}

\newtheorem{theorem}{Theorem}[section]
\newtheorem{lem}[theorem]{Lemma}

\newtheorem{corollary}[theorem]{Corollary}
\newtheorem{definition}[theorem]{Definition}
\newtheorem{remark}[theorem]{Remark}
\newtheorem{example}[theorem]{Example}

\newtheorem{fact}[theorem]{Fact}

\newcommand{\ds}{\displaystyle}

\newcommand{\R}{\mathbb{R}}

\newcommand{\bi}{\begin{itemize}}
\newcommand{\ei}{\end{itemize}}
\newcommand{\ben}{\begin{enumerate}}
\newcommand{\een}{\end{enumerate}}

\newcommand{\bean}{\begin{eqnarray*} }
\newcommand{\eean}{\end{eqnarray*} }
\newcommand{\bea}{\begin{eqnarray} }
\newcommand{\eea}{\end{eqnarray} }
\newcommand{\nn}{\nonumber}
\newcommand{\ba}{\begin{array} }
\newcommand{\ea}{\end{array} }
\newcommand{\beq}{\begin{equation}}
\newcommand{\eeq}{\end{equation}}

\newcommand{\vect}[2]{\left[\begin{array}{cccccc}
     #1 \\
     #2
   \end{array}
  \right]
  }

\renewcommand\thetheorem{\arabic{section}.\arabic{theorem}}

\newcommand{\mt}{\bm{m}_t}
\newcommand{\xt}{\bm{x}_t}
\newcommand{\xhatt}{\hat{\bm{x}}_t}
\newcommand{\lt}{\bm{\ell}_t}
\newcommand{\lhatt}{\hat{\bm{\ell}}_t}
\newcommand{\et}{\bm{e}_t}
\newcommand{\Pt}{\bm{P}_t}
\newcommand{\Pj}{\bm{P}_{(j)}}
\newcommand{\Pjm}{\bm{P}_{(j-1)}}
\newcommand{\Pjnew}{\bm{P}_{(j),\mathrm{new}}}
\newcommand{\at}{\bm{a}_t}
\newcommand{\atnew}{\bm{a}_{t,\mathrm{new}}}
\newcommand{\I}{\bm{I}}
\newcommand{\Lamt}{\bm{\Lambda}_t}
\newcommand{\Lamtnew}{(\bm{\Lambda}_{t})_\mathrm{new}}
\newcommand{\T}{\mathcal{T}}

\newcommand{\D}{\bm{D}}

\newcommand{\rmnew}{\mathrm{new}}
\newcommand{\new}{\mathrm{new}}
\newcommand{\cs}{\text{cs}}

\newcommand{\Lhat}{\hat{\bm{\ell}}}

\newcommand{\Ltil}{\tilde{\bm{\ell}}}
\newcommand{\Phat}{\hat{\bm{P}}}

\newcommand{\Span}{\operatorname{span}}

\newcommand{\add}{\mathrm{add}}

\newcommand{\rank}{\operatorname{rank}}
\newcommand{\E}{\mathbb{E}}

\newcommand{\calc}{\mathcal{C}}
\newcommand{\full}{\mathrm{full}}
\newcommand{\cov}{\operatorname{Cov}}

\newcommand{\SE}{\mathrm{SE}}

\newcommand{\Ijk}{\mathcal{I}_{j,k}}

\begin{document}

\maketitle

\begin{abstract}
This work studies the problem of sequentially recovering a sparse vector $\xt$ and a vector from a low-dimensional subspace $\lt$ from knowledge of their sum $\mt = \xt + \lt$.
If the primary goal is to recover the low-dimensional subspace where the $\lt$'s lie, then the problem is one of online or recursive robust principal components analysis (PCA).
To the best of our knowledge, this is the first correctness result for online robust PCA.
We prove that if the $\lt$'s obey certain denseness and slow subspace change assumptions,
and the support of $\xt$ changes by at least a certain amount at least every so often,
and some other mild assumptions hold,
then with high probability,
the support of $\xt$ will be recovered exactly, and the error made in estimating $\xt$ and $\lt$
will be small.
An example of where such a problem might arise is in separating a sparse foreground and slowly changing dense background in a surveillance video.
\end{abstract}

\section{Introduction}

Principal Components Analysis (PCA) is a widely used tool for dimension reduction.  Given a matrix of data $\D$, PCA seeks to recover a small number of directions that contain most of the variability of the data.
This is typically accomplished by performing a singular value decomposition (SVD) of $\D$ and retaining the singular vectors corresponding to the largest singular values.
A limitation of this procedure is that it is highly sensitive to outliers in the data set.
Recently there has been much work done to develop and analyze algorithms for PCA that are robust with respect to outliers.
A common way to model outliers is as sparse vectors \cite{error_correction_PCP_l1}.
In seminal papers Cand\`{e}s et. al. and Chandrasekaran et. al. introduced the Principal Components Pursuit (PCP) algorithm and proved its robustness to sparse outliers \cite{rpca}, \cite{rpca2}.
Principal Components Pursuit poses the robust PCA problem as identifying a low rank matrix and a sparse matrix from their sum.
The algorithm is to minimize a weighted sum of the nuclear norm of the low rank matrix and the vector $\ell_1$ norm of the sparse matrix subject to their sum being equal to the observed data matrix.
Stronger results for the PCP program can be found in \cite{chen2011low}.
Other methods such as \cite{outlier_pursuit} model the entire column vector as being either correct or an outlier.  Some other works on the performance guarantees for batch robust PCA include \cite{rpca_tropp}, \cite{linear_inverse_prob}, and \cite{noisy_undersampled_yuan}.
All of these methods require waiting until all of the data has been acquired before performing the optimization.

In this work we consider an online or recursive version of the robust PCA problem where we seek to separate vectors into low dimensional and sparse components as they arrive, using the previous estimates, rather than re-solving the entire problem at each time $t$.  An application where this type of problem is useful is in video analysis \cite{Torre03aframework}.  Imagine a video sequence that has a distinct background and foreground.  An example might be a surveillance camera where a person walks across the scene.  If the background does not change very much, and the foreground is sparse (both practical assumptions), then separating the background and foreground can be viewed as a robust PCA problem.  Sparse plus low rank decomposition can also be used to detect anomalies in network traffic patterns \cite{mardani2013dynamic}.  In all such an applications an online solution is desirable.

\subsection{Contributions}

To the best of our knowledge, this is among the first works that provides a correctness result for an online (recursive) algorithm for sparse plus low-rank matrix recovery.
We study the ReProCS algorithm introduced in \cite{ReProCS_IT}.
As shown in \cite{han_tsp}, with practical heuristics used to set its parameters, ReProCS has significantly improved recovery
performance compared to other recursive (\cite{GRASTA,mateos2012robust, mardani2013dynamic}) and even batch methods (\cite{rpca, Torre03aframework, mateos2012robust}) for many simulated and real video datasets.


Online algorithms are needed for real-time applications such as video surveillance or for other streaming video applications. Moreover, even for offline applications, they are faster and need less storage compared to batch techniques.  Finally, as we will see, online approaches can provide a way to exploit temporal dependencies in the dataset (in this case, slow subspace change), and we use this to
allow for more correlated support sets of the sparse vectors than do the various results for  PCP \cite{rpca,rpca2,chen2011low}.
Of course this comes at a cost.
Our result needs a tighter bound on the rank-sparsity product as well as an accurate estimate of the initial low dimensional subspace and analyzes an algorithm that requires knowledge of subspace change model parameters.

We show that as long as algorithm parameters are set appropriately (which requires knowledge of the subspace change model parameters), a good-enough estimate of the initial subspace is available, a slow subspace change assumption holds, the subspaces are dense enough for a given maximum support size and a given maximum rank, and there is a certain amount of support change at least every so often, then the support can be exactly recovered with high probability.
Also the sparse and low-rank matrix columns can be recovered with bounded and small error.

Partial results have been provided for ReProCS in the recent work of Qiu et. al. \cite{ReProCS_IT} and follow up papers \cite{zhan_correlated, rrpcp_globalsip};
however, none of these results is a correctness result.
All require an assumption that depends on intermediate algorithm estimates.
Recent work of Feng et. al. from NIPS 2013 \cite{OnlinePCA_ContaminatedData}, \cite{rpca_stochatistic_optimization} provides partial results for online robust PCA.
One of these papers, \cite{OnlinePCA_ContaminatedData}, does not model the outlier as a sparse vector;  \cite{rpca_stochatistic_optimization} does, but it again contains a partial result.
Moreover the theorems in both papers only talk about asymptotically converging to the solution of the batch problem, whereas in the current work,
we exploit slow subspace change to actually relax a key assumption needed by the batch methods (that of uniformly distributed random supports or of very frequent support change).
A more detailed comparison of our results with \cite{rpca}, \cite{rpca2}, \cite{rpca_stochatistic_optimization}, \cite{mod_pcp}, and \cite{hsu2011robust} is given in Section \ref{comparison}.
Other work only provides an algorithm without proving any performance results; for example \cite{GRASTA}, \cite{mateos2012robust}.

Our result uses the overall proof approach of \cite{ReProCS_IT} as its starting point.
The new techniques used in proving our correctness result are described in Section \ref{pfsketch}.
Moreover, new techniques are also needed to prove that various practical models for support change follow the support change assumptions needed by our main result (Section \ref{motion}).

\subsection{Notation}

We use lowercase bold letters for vectors, capital bold letters for matrices, and calligraphic capital letters for sets.
We use $\bm{x}'$ for the transpose of $\bm{x}$.
The  2-norm of a vector and the induced 2-norm of a matrix are denoted by $\| \cdot \|_2$.
We refer to a matrix with orthonormal columns as a {\em basis matrix}.
Notice that for a basis matrix $\bm{P}$, $\bm{P}'\bm{P} = \bm{I}$.
 For a set $\mathcal{T}$ of integers, $|\mathcal{T}|$ denotes its cardinality.
For a vector $\bm{x}$, $\bm{x}_{\mathcal{T}}$ is a vector containing the entries of $\bm{x}$ indexed by $\mathcal{T}$.
Define $\I_{\mathcal{T}}$ to be an $n \times |\mathcal{T}|$ matrix of those columns of the identity matrix indexed by $\mathcal{T}$.
Then let $\bm{A}_{\mathcal{T}} := \bm{AI}_{\mathcal{T}}$.
We use the interval notation $[a, b]$ to mean all of the integers between $a$ and $b$, inclusive, and similarly for $(a,b)$ etc.
For a matrix $\bm{A}$, the restricted isometry constant (RIC) $\delta_s(\bm{A})$ is the smallest real number $\delta_s$ such that
\[
(1-\delta_s)\|\bm{x}\|_2^2 \leq \|\bm{Ax}\|_2^2 \leq (1+\delta_s) \|\bm{x}\|_2^2
\]
for all $s$-sparse vectors $\bm{x}$ \cite{candes_rip}.  A vector $\bm{x}$ is $s$-sparse if it has $s$ or fewer non-zero entries.
For Hermitian matrices $\bm{A}$ and $\bm{B}$, the notation $\bm{A}\preceq\bm{B}$ means that $\bm{B}-\bm{A}$ is positive semi-definite.
For a Hermitian matrix $\bm{H}$, $\bm{H}\overset{\mathrm{EVD}}= \bm{U\Lambda U}'$ denotes its eigenvalue decomposition.  Similarly for any matrix $\bm{A}$, $\bm{A}\overset{\mathrm{SVD}}=\bm{U\Sigma V'}$ denotes its singular value decomposition.

\subsection{Organization}
The rest of the paper is organized as follows.
In Section \ref{Problem Definition and Assumptions} we describe the signal model assumed for our data.
Our main result (Theorem \ref{thm1}) is presented in Section \ref{results},
and the proof is given in Section \ref{pfsketch}.
Section \ref{motion} describes examples of motion that satisfy our support change assumptions, and contains Corollary \ref{probcor} which is a result for a simple example of motion that satisfies our support change assumptions.  The corollary is proved in Section \ref{probpf}.
A discussion of our results and support change model along with comparisons with other works can be found in Section \ref{comparison}.
We show a simulation experiment that demonstrates our result in Section \ref{simulation}.
Finally, some concluding remarks and directions for future work are given in Section \ref{conclusions}.


\section{Problem Definition and Assumptions}\label{Problem Definition and Assumptions}

At time $t$ we observe a vector $\mt \in \R^n$ that is the sum of a vector from a slowly changing low-dimensional subspace $\lt$ and a sparse vector $\xt$.
So
\[
\mt = \lt + \xt \qquad \text{ for } t = 0, 1, 2, \dots, t_{\max},
\]
with the possibility that $t_{\max} = \infty$.  We model the low-dimensional $\lt$'s as $\lt = \Pt \at$ for a basis matrix $\Pt$ that is allowed to change slowly over time.
Given an estimate of the initial subspace $\hat{\bm{P}}_{(0)}$,
the goal is to obtain estimates $\xhatt$ and $\lhatt$ at each time $t$ and to periodically update the estimate of $\Pt$.

\subsection{Model on \texorpdfstring{$\lt$}{lt}}\label{ltmodel}

\begin{enumerate}

\item Subspace Change Model for $\lt$

Let $t_j$ for $j = 1, \dots, J$ be the times at which the subspace where the $\lt$'s lie changes.
We assume $\lt = \Pt \at$ where $\Pt = \Pj$ for $t_{j} \leq t < t_{j+1}$.

$\Pj$ is a basis matrix that changes as $\Pj = [ \Pjm \ \Pjnew ]$.
Because $\Pj$ is a basis matrix, $ \Pjm \perp \Pjnew$.
Let $r_j = \rank(\Pj)$ and define $r := r_J = \max_{j} \rank{\Pj}$.
Also let $c_{j,\rmnew} = \rank(\Pjnew)$ and define $c := \max_{j}\rank(\Pjnew)$.
Observe that
\[
r = \rank ( [ \bm{\ell}_1, \dots, \bm{\ell}_{t_{\max}} ] ) \quad\text{ and }\quad r \leq r_0 + Jc .
\]

For $t\in[t_{j}, t_{j+1})$, $\lt$ can be written as $\lt = [ \Pjm \ \Pjnew ]\vect{\bm{a}_{t,*}}{\atnew}$.
where
\[
\bm{a}_{t,*} := {\Pjm}'\lt \quad\text{ and }\quad \bm{a}_{t,\rmnew} := {\Pjnew}'\lt
\]

\item Assumptions and notation for $\at$

We assume that the $\at$'s are zero mean bounded random variables that are mutually independent over time.  Let
\[
\gamma := \sup_t \|\at\|_{\infty}
\quad\text{ and }\quad
\gamma_{\rmnew} := \sup_{t} \|\atnew\|_{\infty}.
\]
Define $\Lamt := \cov(\at)$ and assume it is diagonal.
Let $\Lamtnew := \cov(\atnew)$.
Define $\lambda^- := \inf_t \lambda_{\min}(\Lamt)$ and $\lambda^+ := \sup_t \lambda_{\max}(\Lamt)$,
and assume that $0 < \lambda^- \leq \lambda^+ < \infty$.
Also, for an integer $d$, define $\lambda^-_{\new} := \min_{j} \min_{t\in[t_j,t_j+d]}\lambda_{\min}((\bm{\Lambda}_{t})_{\new})$ and
$\lambda^+_{\new} := \max_{j} \max_{t\in[t_j,t_j+d]} \lambda_{\max}((\bm{\Lambda}_{t})_{\new})$.

Then define
\[
f := \frac{\lambda^+}{\lambda^-} \quad \text{ and } \quad g := \frac{\lambda^+_{\rmnew}}{\lambda^-_{\rmnew}}.
\]

Notice that $f$ is a bound on the condition number of $\bm{\Lambda}_t$.  And $g$ is a bound on the condition number of  $\Lamtnew$ for the first $d$ time instants after a subspace change.

\end{enumerate}

\subsection{Model on \texorpdfstring{$\xt$}{xt}} \label{xtmodel}

Let $\T_t := \{ i : (\xt)_i \neq 0\}$ be the support set of $\xt$ and let $s = \max_{t}| \T_t |$ be the size of the largest support.
Let $x_{\min} := \inf_t \min_{i\in\T_t} |(\xt)_i|$ denote the size of the smallest non-zero entry of any $\xt$.


Divide the interval $[1,t_{\max}]$ into subintervals of length $\beta$. For the interval $[(u-1)\beta,u\beta-1]$ let $\mathcal{T}_{(i),u}$ for $i = 1,\dots, l_u$ be mutually disjoint subsets of $\{ 1, \dots, n\}$ such that for every $t \in [(u-1)\beta,u\beta-1]$,

\begin{equation}\label{union}
\T_t \subseteq \mathcal{T}_{(i),u} \cup \mathcal{T}_{(i+1),u} \quad\text{ for some } i.
\end{equation}

Then define
\begin{align}\label{alphabyp}
h(\beta) &:= \max_{u = 1,\dots, \lceil\frac{t_{\max}}{\beta}\rceil}\max_{i}\big| \{ t \in [(u-1)\beta+1,u\beta] \ : \ \T_t \subseteq \mathcal{T}_{(i),u}\cup \mathcal{T}_{(i+1),u}\} \big|
\end{align}
First notice that \eqref{union} can always be trivially satisfied by choosing $l_u =1$ and $\mathcal{T}_{(1),u} = \{  1, \dots, n \}$.
Also notice that $h(\beta)$ is a bound on how long the support of $\xt$ remains in a given area during the intervals $[(u-1)\beta,u\beta-1]$.  We should also point out that as defined above, $h(\beta)$ depends on the choice of $\mathcal{T}_{(i),u}$.  The trivial choice will give $h(\beta) = \beta$.  To eliminate this dependence, we could specify that the $\mathcal{T}_{(i)}$ be chosen optimally.  So we could instead define
\begin{align}\label{hstar}
h^*(\beta) &:=   \max_{u = 1,\dots, \lceil\frac{t_{\max}}{\beta}\rceil}  \min_{l_u=1,\dots,n} \min_{\substack{\text{mutually disjoint} \\ \mathcal{T}_{(1),u},\dots,\mathcal{T}_{(l_u),u} \text{satisfying} \eqref{union} } }\max_{i}\big| \{ t \in [(u-1)\beta,u\beta-1] \ : \ \T_t \subseteq \mathcal{T}_{(i),u}\cup \mathcal{T}_{(i+1),u}\} \big|
\end{align}
Observe that $h^*(\beta)$ is $h(\beta)$ with the $\mathcal{T}_{(i)}$ chosen to minimize $h(\beta)$.

As we will see, we need an upper bound on $h^*(\beta)$ for large enough $\beta$.
Practically, this means that there is some support change every so often.
We show several examples of this model in Section \ref{motion}.

\begin{remark}\label{dsubset}
We could replace \eqref{union} above with
\[
\T_{t}\subseteq \T_{(i),u} \cup \T_{(i+1),u} \cup \dots \cup \T_{(i+d-1),u}.
\]
In this case we would need a tighter bound on $h^*(\beta)$ because the conclusion of Lemma \ref{blockdiag} will be
$\|\bm{M}\|_2 \leq 2d^2\sigma^+ h^*(\beta) $. The proof is easily modified by adding additional bands around the block diagonal.  We choose to use $d=2$ for purposes of exposition.
\end{remark}

\subsection{Subspace Denseness}

Below we give the definition of the denseness coefficient $\kappa_s$.
\begin{definition}
For a basis matrix $\bm{P}$, define $\ds\kappa_s(\bm{P}) := \max_{|\mathcal{T}| \leq s} \| {\I_{\mathcal{T}}}' \bm{P} \|_2$.
\end{definition}

As described in \cite{ReProCS_IT}, $\kappa_s$ is a measurement of the denseness of the vectors in the subspace $\operatorname{range}(\bm{P})$.  Notice that \emph{small} $\kappa_s$ means that the columns of $\bm{P}$ are dense vectors.
The reason for quantifying denseness using $\kappa_s$ is the following lemma from \cite{ReProCS_IT}.

\begin{lem}\label{kappadelta}
For a basis matrix $\bm{P}$, $\delta_s(\I - \bm{P}\bm{P}') = \left(\kappa_s(\bm{P})\right)^2$.
\end{lem}

Lemma \ref{kappadelta} says that if the columns of $\bm{P}$ are dense,
then the orthogonal projection onto the orthogonal complement of the range of $\bm{P}$ will have a small RIC.  Our result assumes a bound on $\kappa_s(\bm{P}_{(J)})$ and a tighter bound on $\kappa_s(\bm{P}_{(j),\new})$.

\section{Main Result}\label{results}

\subsection{Main Result}

In this section we state and discuss our main result for the ReProCS algorithm introduced in \cite{ReProCS_IT}.
We restate the ReProCS algorithm as Algorithm \ref{reprocs} and briefly explain its main idea in Section \ref{algosubsec}.

\begin{theorem}[Correctness result for Algorithm \ref{reprocs} under the model given in Section \ref{Problem Definition and Assumptions}]\label{thm1}
\

Pick a $\zeta$ that satisfies
\[
\zeta  \leq  \min\left(\frac{10^{-4}}{r^2},\frac{1.5 \times 10^{-4}}{r^2 f},\frac{1}{r^{3}\gamma^2}\right).
\]
Suppose
\begin{enumerate}

\item $\ds\|(\I - \Phat_{(0)} \Phat_{(0)}{}') \bm{P}_{(0)}\|_2 \le r_0 \zeta$;

\item The algorithm parameters are set as:
\begin{itemize}
\item $K = \left\lceil \frac{\log(0.17c\zeta)}{\log(0.72)}\right\rceil$;
\item $\xi = \sqrt{c}\gamma_{\rmnew} + \sqrt{\zeta}(\sqrt{r} + \sqrt{c})$;
\item $7\xi \leq \omega \leq x_{\min} - 7 \xi$;
\item $\alpha = C (\log(6KJ) + 11\log(n))$ for a constant
$C \geq C_{\add}$ with
\[
 C_{\add} := \frac{4800}{(\zeta \lambda^-)^2}\max\{16,(1.2\xi)^4\}
\]
\end{itemize}
\item The subspace changes slowly enough such that the model parameters satisfy:
\begin{itemize}
\item $t_{j+1} - t_{j} > d \geq K\alpha \ $ for all $j$;
\item $r < \min\{ n , t_{j+1} - t_j \}$ for all $j$;
\item $\sqrt{c}\gamma_{\rmnew} + \sqrt{\zeta}(\sqrt{r} + \sqrt{c}) \leq \frac{x_{\min}}{14}$;
\item $g \leq \sqrt{2}$;
\end{itemize}
\item \label{supchass}The support of $\xt$ changes enough such that for the $\alpha$ chosen above,
\[
h^*(\alpha) \leq h^+ \alpha 
\]
where $h^+ := \frac{1}{200}$.

\item The low dimensional subspace is dense such that
\begin{itemize}
\item $\kappa_{2s}(\bm{P}_{(J)}) \leq 0.3$;
\item  $\max_{j}\kappa_{2s}(\Pjnew) \leq 0.02$.
\end{itemize}
\end{enumerate}

Then, with probability at least $1 - n^{-10}$, at all times $t$

\begin{enumerate}

\item The support of $\xt$ is recovered exactly, i.e. $\hat\T_t = \T_t$

\item The recovery error satisfies:
\begin{align*}
\|\hat{\bm{x}}_t - \xt \|_2 &\leq
\begin{cases}
1.2 \left(1.83\sqrt{\zeta} +  (0.72)^{k-1}\sqrt{c}\gamma_{\rmnew}  \right) &   t \in[t_j + (k-1)\alpha, t_j + k\alpha -1], \ k=1,2, \dots, K \\
2.4\sqrt{\zeta} & t \in [t_j + K \alpha, t_{j+1} - 1]
\end{cases}
\end{align*}

\item The subspace error
$\SE_t := \|( \I - \hat{\bm{P}}_t \hat{\bm{P}}_t{}' ) \Pt \|_2$ satisfies:
\begin{align*}
\SE_t \leq
\begin{cases}
 10^{-2} \sqrt{\zeta} +  0.72^{k-1} & t \in [t_j + (k-1)\alpha, t_j + k\alpha -1], \ k=1,2, \dots, K  \\
10^{-2} \sqrt{\zeta}   &  t \in [t_j + K \alpha, t_{j+1} - 1].
\end{cases}
\end{align*}

\end{enumerate}

\end{theorem}

\begin{corollary}
In the special case where $(\bm{\Lambda}_{t})_{\new} $ is a multiple of the identity, that is $(\bm{\Lambda}_{t})_{\new} = \lambda_{t,\new}\bm{I}$, then the same result as above can be proven with $h^+ = \frac{1}{88}$.  This assumption would hold for example if only one new direction is added at a time i.e. $c=1$.  For this proof, the Cauchy-Schwarz inequality is not needed, so there will not be a square root in the expression for $b_{4,k}$ on page \pageref{b4}
\end{corollary}

Theorem \ref{thm1} says that if an accurate estimate of the initial subspace is available, the algorithm parameters are set appropriately, the low-dimensional subspace changes slowly enough, the support of the sparse part changes quickly enough, and the low-dimensional subspace is dense, then with high probability the support of the sparse vector will be recovered exactly and the error in estimating both $\xt$ and $\lt$ will be small at all times $t$.  Also, the error in estimating the low-dimensional subspace will be initially large when new directions are added, but decays exponentially to a small constant times $\sqrt{\zeta}$.

The assumptions on the support change of $\xt$ in Theorem \ref{thm1} are stronger than needed. We presented the model as such for simplicity and ease of exposition.  Our proof only uses the support change assumptions for the intervals $[t_j+(k-1)\alpha,t_j+k\alpha-1]$ for $j = 1,\dots,J$ and $k=1,\dots,K$.  These are the intervals where projection-PCA is done (see Algorithm \ref{reprocs}).



\subsection{The ReProCS Algorithm}\label{algosubsec}

The ReProCS algorithm presented here was introduced in \cite{ReProCS_IT}.  A more practical version including heuristics for setting the parameters was given in \cite{han_tsp}.  The basic idea of ReProCS is as follows.
Given an accurate estimate of the subspace where the $\lt$'s lie, projecting the measurement $\mt = \xt + \lt$ onto the orthogonal complement of the estimated subspace will nullify most of $\lt$.  The denseness of $\lt$ implies that this projection will have small RIC (Lemma \ref{kappadelta}) so the sparse recover step will produce an accurate estimate $\hat{\bm{x}}_t$.  Then, subtraction also gives a good estimate $\hat{\bm{\ell}}_t = \mt - \hat{\bm{x}}_t$.  Using these $\hat{\bm{\ell}}_t$, the algorithm successively updates the subspace estimate by a modification of the standard PCA procedure, which we call projection PCA.

\begin{algorithm}[!ht]
\caption{Recursive Projected CS (ReProCS) \cite{ReProCS_IT}}\label{reprocs}
{\em Parameters}:  algorithm parameters: $\xi$, $\omega$, $\alpha$, $K$, model parameters: $t_j$, $c_{j,\rmnew}$ \\
{\em Input}:  $\bm{m}_t$, \\
{\em Output}:  $\xhatt$, $\Lhat_t$, $\Phat_{t}$ \\

Set $\Phat_{t} \leftarrow \Phat_{(0)}$,  $j \leftarrow 1$, $k\leftarrow 1$.

For $t >0$, do the following:
\ben
\item Estimate $\mathcal{T}_t$ and $\bm{x}_t$ via Projected CS:
\ben
\item \label{othoproj} Nullify most of $\bm{\ell}_t$: set $\bm{\Phi}_{t} \leftarrow \bm{I} - \Phat_{t-1} \Phat_{t-1}{}'$, compute $\bm{y}_t \leftarrow \bm{\Phi}_{t} \bm{m}_t$
\item \label{Shatcs} Sparse Recovery: compute $\hat{\bm{x}}_{t,\cs}$ as the solution of $\min_{\bm{x}} \|\bm{x}\|_1 \ s.t. \ \|\bm{y}_t - \bm{\Phi}_{t} \bm{x}\|_2 \leq \xi$
\item \label{That} Support Estimate: compute $\hat{\mathcal{T}}_t = \{i: \ |(\hat{\bm{x}}_{t,\cs})_i| > \omega\}$
\item \label{LS} LS Estimate of $\bm{x}_t$: compute $(\hat{\bm{x}}_t)_{\hat{\mathcal{T}}_t}= ((\bm{\Phi}_t)_{\hat{\mathcal{T}}_t})^{\dag} \bm{y}_t, \ (\hat{\bm{x}}_t)_{\hat{\mathcal{T}}_t^{c}} = 0$
\een
\item Estimate $\bm{\ell}_t$: $\hat{\bm{\ell}}_t = \bm{m}_t - \hat{\bm{x}}_t$.
\item \label{PCA}
Update $\Phat_{t}$: K Projection PCA steps.
\ben

\item If $t = t_j + k\alpha-1$,
\ben

\item $\Phat_{(j),\rmnew,k} \leftarrow$ proj-PCA$\left(\left[\hat{\bm{\ell}}_{t_j+(k-1)\alpha}, \dots, \hat{\bm{\ell}}_{t_j+k\alpha-1}\right],\Phat_{(j-1)},c_{j,\rmnew}\right)$.

\item set $\Phat_{t} \leftarrow [\Phat_{(j-1)} \ \Phat_{(j),\rmnew,k}]$; increment $k \leftarrow k+1$.
\een
Else set $\Phat_{t} \leftarrow \Phat_{t-1}$.

\item If $t = t_j + K \alpha - 1$, then set $\Phat_{(j)} \leftarrow [\Phat_{(j-1)} \ \Phat_{(j),\rmnew,K}]$.

Increment $j \leftarrow j + 1$. Reset $k \leftarrow 1$.
\een
\item Increment $t \leftarrow t + 1$ and go to step 1.
\een

Where,

$\bm{Q} \leftarrow \text{proj-PCA}(\bm{D},\bm{P},r)$
\ben
\item Projection: compute $\bm{D}_{\mathrm{proj}} \leftarrow (\bm{\bm{I} - P P}') \bm{D}$
\item PCA: compute $\frac{1}{\alpha}  \bm{D}_{\mathrm{proj}}{\bm{D}_{\mathrm{proj}}}' \overset{\mathrm{EVD}}{=}
\left[ \begin{array}{cc}\bm{Q} & \bm{Q}_{\perp} \\\end{array}\right]
\left[ \begin{array}{cc} \bm{\Lambda} & 0 \\0 & \bm{\Lambda}_{\perp} \\\end{array}\right]
\left[ \begin{array}{c} \bm{Q}' \\ {\bm{Q}_{\perp}}'\\\end{array}\right]$
where $\bm{Q}$ is an $n \times r$ basis matrix and  $\alpha$ is the number of columns in $\bm{D}$.
\een

\end{algorithm}

\section{Support Change Examples}\label{supchex}\label{motion}
The support change model assumed in Theorem \ref{thm1} is quite general, but also rather abstract.  So in this section, we provide some concrete examples of motion patterns that will satisfy our model.

We introduce the following definition for ease of notation.

\begin{definition}
Define the interval
\[
\mathcal{J}_u := [(u-1)\alpha,u\alpha-1]
\]
for $u = 1,\dots,\lceil\frac{t_{\max}}{\alpha}\rceil$.
\end{definition}

\subsection{Support change by at least $\frac{s}{2}$ every so often.}

\begin{figure}[ht]
\includegraphics[width= \columnwidth]{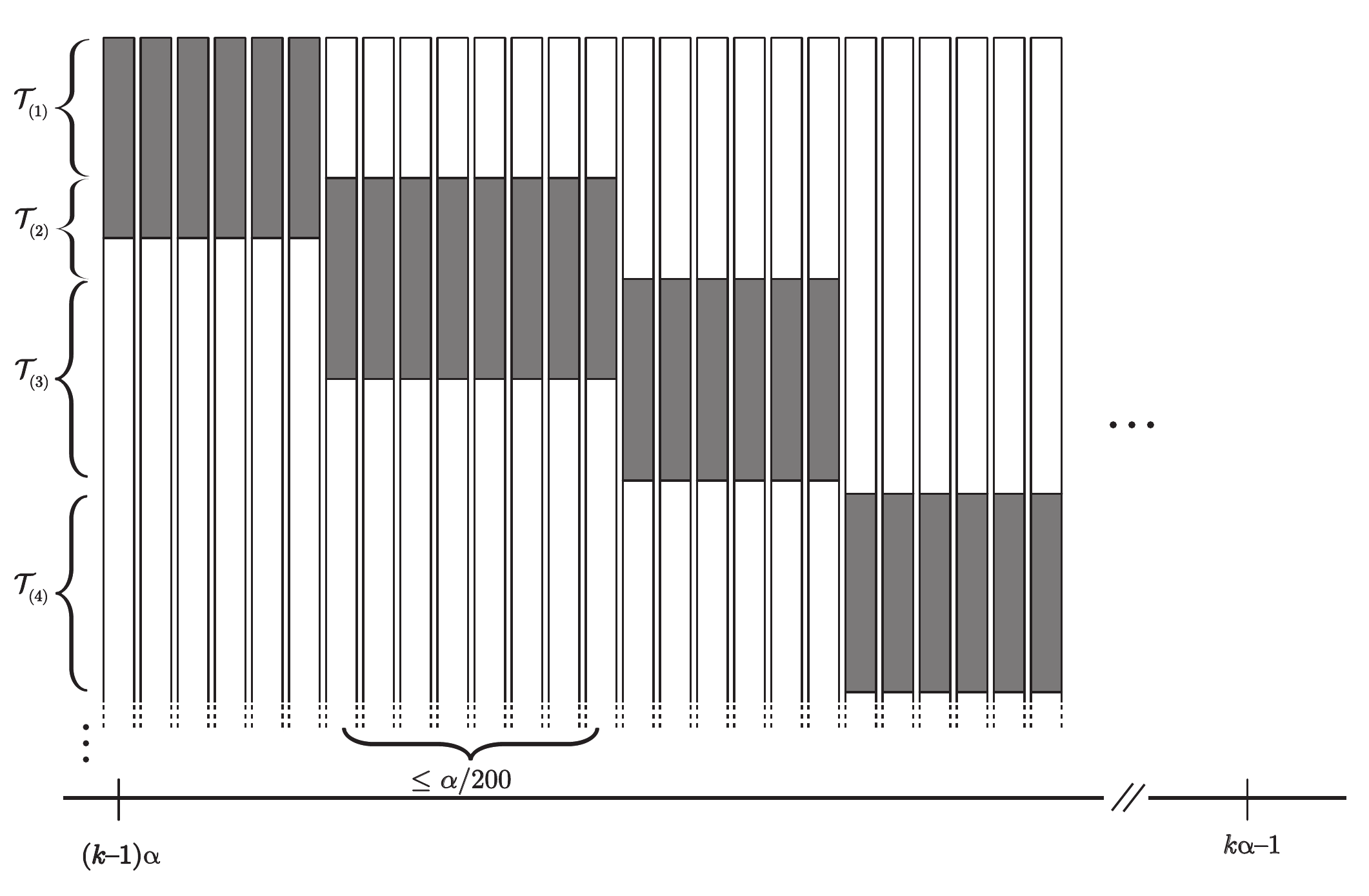}
\caption{\label{supportfig} A visual representation of Example \ref{sby2} and its proof.}
\end{figure}

\begin{example}
\label{sby2}
Consider one dimensional motion where the support moves down the vector.
Also assume that when the support reaches $n$, it starts over at 1.
Let $b_{u}$ be the number of times the support of $\xt$ changes during the interval $\mathcal{J}_{u}$.
If during the interval $\mathcal{J}_{u}$,
\begin{enumerate}
\item $b_{u} \leq \frac{n}{2s}$ ; \label{bupbnd}
\item  the support changes at least once every $\frac{\alpha}{200}$ time instants ;\label{alphaby44}
\item when the support of $\xt$ changes, it moves down by at least $\frac{s}{2}$ indices but fewer than $2s$ indices ; \label{mvmt}
\end{enumerate}
then, $h^*(\alpha) \leq \frac{\alpha}{200}$.
Notice that \ref{alphaby44}) implies that $b_{u}\geq 199$.
\end{example}

To help clarify Example \ref{sby2},
consider
\begin{align*}
\mathcal{T}_{t} &= \{ 1,\dots,10 \} \quad \text{ for the first $\frac{\alpha}{200}$ or fewer frames}, \\
\mathcal{T}_{t} &= \{ 6,\dots,15\}\quad\text{ for the next $\frac{\alpha}{200}$ or fewer frames} ,\\
&\text{and so on.}
\end{align*}

\begin{proof}

First notice that together \ref{bupbnd}) and \ref{mvmt}) ensure that during $\mathcal{J}_{u}$, the object does not revisit indices that it has previously occupied.
If for each interval $\mathcal{J}_{u}$, we can construct one set of $\mathcal{T}_{(i),u}$'s so that $\big| \{ t \in [(u-1)\alpha,u\alpha-1] \ : \ \T_t \subseteq \mathcal{T}_{(i),u}\cup \mathcal{T}_{(i+1),u}\} \big| \leq \frac{\alpha}{200}$, we will have shown that for these $\mathcal{T}_{(i),u}$, $h(\alpha)\leq \frac{\alpha}{200}$.  Because $h^*(\alpha)$ takes the minimum over choices of $\mathcal{T}_{(i)}$, we will be done.

Each interval $[(u-1)\alpha,u\alpha-1]$ is treated in the same way, so we remove the subscript $u$ for simplicity.
Let $\mathcal{T}^{[j]}$ for $j = 1, \dots ,m$ be the distinct supports during the interval $\mathcal{J}_u$.  That is $\mathcal{T}_t = \mathcal{T}^{[j]}$ for some $j$.
Now define $\mathcal{T}_{(i)} = \mathcal{T}^{[i]} \setminus \mathcal{T}^{[i+1]}$ for $i = 1, \dots, m-1$,
and $\mathcal{T}_{(m)} = \mathcal{T}^{[m]}$.
Because the support moves by at least $\frac{s}{2}$ indices, $\mathcal{T}^{[j]}$ only has overlap with $\mathcal{T}^{[j-1]}$ and $\mathcal{T}^{[j+1]}$ (In particular, $\mathcal{T}^{[j]}\subseteq \big(\mathcal{T}^{[j+2]}\big)^c$ ).
Thus the $\mathcal{T}_{(i)}$ are disjoint.
Next we show that $\mathcal{T}^{[j]} \subseteq \mathcal{T}_{(j)}\cup \mathcal{T}_{(j+1)}$.
\begin{align*}
\mathcal{T}^{[j]} &= \left(\mathcal{T}^{[j]} \cap \big(\mathcal{T}^{[j+1]}\big)^c \right) \cup \left( \mathcal{T}^{[j]} \cap \mathcal{T}^{[j+1]} \right)\\
&= \mathcal{T}_{(j)} \cup \left( \mathcal{T}^{[j]} \cap \mathcal{T}^{[j+1]} \right) \\
&\subseteq \mathcal{T}_{(j)} \cup \left( \big(\mathcal{T}^{[j+2]}\big)^c \cap \big(\mathcal{T}^{[j+1]}\big) \right) \\
&\subseteq \mathcal{T}_{(j)} \cup \mathcal{T}_{(j+1)}
\end{align*}
Finally, by \ref{alphaby44}), the support changes at least once every $\frac{\alpha}{200}$ time instants, so with the $\mathcal{T}_{(i)}$ chosen this way, $h(\alpha)$, defined in \eqref{alphabyp}, will be less than $\frac{\alpha}{200}$.  Since $h^*(\alpha)$ is the minimum over all choices of $\mathcal{T}_{(i)}$, $h^*(\alpha)\leq \frac{\alpha}{200}$.
\end{proof}

\begin{remark}
As per Remark \ref{dsubset}, we can replace $\frac{s}{2}$ above with $\frac{s}{d}$.  In many cases this is a more realistic assumption.  As noted previously, this will require a smaller bound on $h^*(\alpha)$.  For example, if $d=5$, then we would need $h^*(\alpha)\leq\frac{\alpha}{1250}$, or if $d=10$, then $h^*(\alpha)$ would have to be less than $\frac{\alpha}{5000}$.
\end{remark}

\begin{example}[Probabilistic model on support change]\label{probsupch}

Again consider one-dimensional motion of the support of $\xt$,
and let $o_t$ be its center at time $t$.
Suppose that the support moves according to the model
\begin{equation}\label{probsupp}
o_t = o_{t-1} + \theta_t ( 0.6s + \nu_t )
\end{equation}
where $\nu_t$ is Gaussian $\mathcal{N}(0,\sigma^2)$ and $\theta_t$ is a Bernoulli random variable that takes the value 1 with probability $q$ and 0 with probability $1-q$.  Again, when the object reaches $n$, it starts over at 1.  Assume that $\{ \nu_t \}$, $\{ \theta_t \}$ are mutually independent and independent of $\{ \bm{a}_t\}$ for $t = 1,\dots,t_{\max}$.  So according to this model, at each time instant the object moves with probability $q$ and remains stationary with probability $1-q$.  Also, when the object moves, it moves by $0.6s$ plus some noise.  This noise can be viewed as a random acceleration.
\end{example}


\begin{lem}\label{problem}
Under the model of Example \ref{probsupch}, if
\begin{enumerate}[(i)]


\item the length of the data record, $t_{\max} \leq n^{10}$;

\item $\ds\sigma^2 =\frac{\rho s^2}{\log(n)}$ for a $\ds\rho \leq \frac{1}{4000}$;




\item $q \geq 1 - \left(  \frac{n^{-10}}{2(t_{\max}+\alpha)} \right)^{\frac{200}{\alpha}}  $ ;

\item $s \leq \frac{n}{2 \alpha}$ \label{sbnd};


\end{enumerate}
then, the sequence of supports $\mathcal{T}_t$ will satisfy Example \ref{sby2} with probability at least $1 - n^{-10}$.
\end{lem}
The proof is given in Section \ref{probpf}.  Notice that if $s$ grows faster than $\sqrt{\log(n)}$, then the noise also grows with $n$.

\begin{corollary}\label{probcor}
If the supports of the $\xt$ obey the conditions of Lemma \ref{problem} and all other assumptions of Theorem \ref{thm1} are satisfied, then all conclusions will also hold with probability at least $1 - n^{-10}$.
\end{corollary}


\subsection{Other Generalizations}

\begin{example}\label{severy2}
An easy generalization of Example \ref{sby2} is if 1) and 2) stay the same and 3) is relaxed to:
\begin{enumerate}
\item[3*)] the support of $\xt$ moves between $s$ and $4s$ indices over any two consecutive support changes.
\end{enumerate}

In this case the distinct supports $\mathcal{T}^{[j]}$ only have overlap with $\mathcal{T}^{[j-1]}$ and $\mathcal{T}^{[j+1]}$.  Notice that the proof of Example \ref{sby2} only uses this overlap assumption, and at least $\frac{s}{2}$ indices merely ensures that $\mathcal{T}^{[j]}$ only overlaps with $\mathcal{T}^{[j-1]}$ and $\mathcal{T}^{[j+1]}$.
This  example replaces the requirement of motion by at least $\frac{s}{2}$ and at most $2s$ {\em every} time to only require a total between $s$ and $4s$ in two changes.
\end{example}

We should point out that the above examples allow the object to change in size over time as long as its {\em center} moves by the required amount and the object's size is bounded by $s$.

\begin{example}[Disjoint Supports]\label{disjointsupp}
Again assume 1) and 2) from Example \ref{sby2}.
Further assume that in every interval $\mathcal{J}_{u}$, the distinct supports of the $\xt$ are disjoint.  That is for $t \in \mathcal{J}_{u}$, $\mathcal{T}_t = \mathcal{T}_{(i),u}$ for some $i$
and $\mathcal{T}_{(i),u}\cap \mathcal{T}_{(i'),u} = \emptyset$ for $i\neq i'$.
Then the support change assumptions of Theorem \ref{thm1} are satisfied.
\end{example}

In fact, for disjoint supports we can modify Lemma \ref{blockdiag} to $\|\bm{M}\|_2 \leq 2\sigma^+ h^*(\beta)$.
To prove this simply notice that the off diagonal blocks in \eqref{blockstruct} become zero.
Using this we could relax assumption \ref{supchass} in Theorem \ref{thm1} to:
\[
h^*(\alpha) \leq \frac{\alpha}{50}.
\]

\begin{example}[Motion Every Frame]\label{motioneveryframe}
Suppose that the support of $\xt$ is of a constant size $s$, consists of consecutive indices, and moves in the same direction by between 1 and $a$ indices at every time $t$.
In this case, $b_{u} = \alpha$.  So in order for the support to not revisit old indices in the interval $\mathcal{J}_{u}$, we also need $\alpha \leq \frac{n}{a}$.
Because the theorem assumes $\alpha = C\log(n)$, choosing $\alpha \leq \frac{n}{a}$ is not very restrictive.

In this case we can let $\mathcal{T}_{(i)} = [(i-1)s + 1,i s]$ which will satisfy \eqref{union}.
Here up to $s$ distinct supports may be contained in $\mathcal{T}_{(i)} \cup \mathcal{T}_{(i+1)}$.
For this choice of $\mathcal{T}_{(i)}$ and $\alpha \leq \frac{n}{a}$, it is easy to see that $h^*(\alpha)\leq s$ (see Figure \ref{everyframe}).  Assumption \ref{supchass} of Theorem \ref{thm1} would then become $s \leq \frac{\alpha}{200}$.  With $\alpha$ as chosen in Theorem \ref{thm1} this would restrict the support size $s$ to being logarithmic in $n$.   
\end{example}

\begin{figure}[!ht]
\centering
\includegraphics[scale=.8]{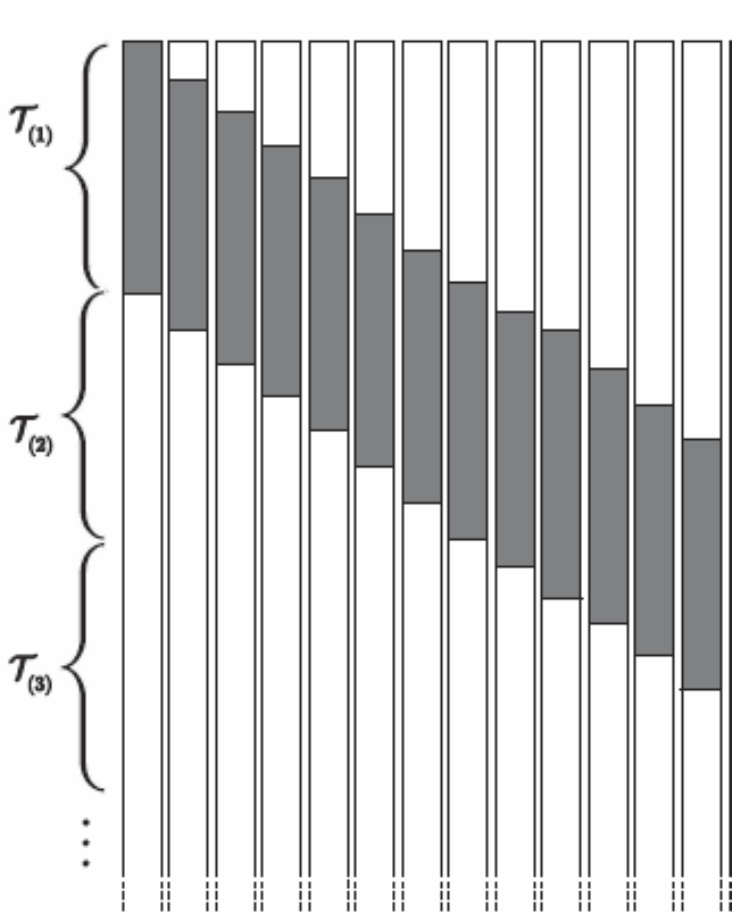}
\caption{Choosing $\mathcal{T}_{(i)}$ when the support moves every frame. \label{everyframe}}
\end{figure}

We make one final remark about our support change assumptions.

\begin{remark}\label{overassumed}
While we assume that the support of $\xt$ changes enough during the intervals $[(u-1)\alpha, u\alpha-1]$ for $u = 1,\dots,\lceil \frac{t_{\max}}{\alpha}\rceil$, our proof only needs the assumption in the intervals where projection-PCA is performed: $[t_j+(k-1)\alpha , t_j + k\alpha-1]$ for $j = 1,\dots, J$ and $k = 1,\dots,K$. Furthermore, our proof does not require anything of the support in the interval $[t_j,t_j + \alpha-1]$ ($k=1$ above).  During this interval we can use subspace denseness to get the bound we need (see Lemma \ref{Dnew0_lem}).
\end{remark}

\section{Discussion of main result and comparison with other works}\label{comparison}

We have proved the main result in Theorem \ref{thm1} and one important corollary for a simple and practical support change model in Corollary \ref{probcor}.
We chose to highlight Example \ref{sby2}, but any of the other examples, Example \ref{severy2}, \ref{disjointsupp}, or \ref{motioneveryframe}. could also be used to get corollaries for Theorem \ref{thm1}.
These results improve upon the result of \cite{ReProCS_IT} by removing the denseness requirements on the quantities $(\bm{I} - \bm{P}_{(j),\rmnew}{\bm{P}_{(j),\rmnew}}')\hat{\bm{P}}_{(j),\rmnew,k}$ and $(\bm{I} - \bm{\hat{P}}_{(j-1)}\bm{\hat{P}}_{(j-1)}{}' -\hat{\bm{P}}_{(j),\rmnew,k}\hat{\bm{P}}_{(j),\rmnew,k}{}')\bm{P}_{(j),\rmnew}$.

First we discuss the data model of Section \ref{Problem Definition and Assumptions}.  This model assumes that after a subspace change, $\|\bm{a}_{t,\rmnew}\|_{\infty}$ and therefore also $\|\Lamtnew\|_2$ are initially small.  After $t_{j}+d$, the eigenvalues of $\Lamtnew$ are allowed to increase up to $\lambda^+$.
Thus a new direction added at time $t_j$ can have variance as large as $\lambda^+$ by $t_{j}+d$.

An important limitation of our results is that we analyze an algorithm (ReProCS) that needs knowledge of some model parameters which is not true of other algorithms such as PCP.
We also require an accurate initial subspace estimate and slow subspace change.
By slow subspace change we mean both that there is a large enough delay between subspace changes,
and that the projection of $\lt$ along the new directions is initially small.
However, as explained and demonstrated in \cite{ReProCS_IT}, both are often valid for video sequences.
Another limiting assumption we make is the zero mean and independence of the $\lt$'s over time.
If a mean background image (obtained by averaging an initial sequence of background only training data) is subtracted from all measurements, then zero mean is valid.
Moreover, if background variation is due to small and random illumination changes, then independence is also valid (or close to valid).
This assumption is used for simplicity, and allows us to apply the matrix Hoeffding inequality.
In \cite{ReProCS_IT} Qiu et. al. suggest that a similar result for a
more realistic autoregressive model on the $\at$'s can be proven using the Azuma inequality (see \cite{zhan_correlated}).
It should be possible to do this for our result as well.


The important assumptions on the low-dimensional vectors are the bounds on $\kappa_{2s}(\bm{P}_J)$ and $\kappa_{2s}(\bm{P}_{(j),\rmnew})$.
The way $\kappa_s$ is defined, these bounds simultaneously place restrictions on denseness of $\lt$, $r = \rank(\bm{P}_J)$, and $s$ (the maximum sparsity of any $\xt$).
To compare our assumptions with those of \cite{rpca}, we could assume $\kappa_1(\bm{P}_{(J)}) \leq \sqrt{\frac{\mu r}{n}}$,
where $\mu$ is any value between $1$ and $\frac{n}{r}$.
It is easy to show that $\kappa_s(\bm{P}) \leq \sqrt{s}\kappa_1(\bm{P})$ \cite{ReProCS_IT}.
So if
\begin{equation}\label{srn}
\frac{2 s r}{n}\leq \mu^{-1}(0.3)^2,
\end{equation}
then our assumption of $\kappa_{2s}(\bm{P}_{(J)}) \leq 0.3$ will be satisfied.
The support change model of Corollary \ref{probcor} requires $s\leq C\frac{n}{\log(n)}$ for a constant $C$ and $J \leq C_1 \log(n)$ for another constant $C_1$.  If we assume $r_0$ and $c$ are constant, then the restriction on $J$ implies that $r \leq r_0 + cJ$ grows as $\log(n)$.  These two assumptions together still imply that $\frac{sr}{n}$ is bounded by a constant (as in \eqref{srn}).  Alternatively, if we use Example \ref{motioneveryframe} to get another corollary of Theorem \ref{thm1}, we will need $s \leq C\log(n)$ and \eqref{srn}.

Up to differences in the constants, \eqref{srn} is the same requirement found in \cite{hsu2011robust} (which studies the PCP program and is an improvement over \cite{rpca2}), except that \cite{hsu2011robust} does not need specific bounds on $s$ and $r$.
The above assumptions on $s$ and $r$ are stronger than those used by \cite{rpca} (which also studies the batch approach PCP).
There $s$ is allowed to grow linearly with $n$, and $r$ is simultaneously allowed to grow as $\frac{n}{\log(n)^2}$.
However, neither comparison is direct because we do not need denseness of the right singular vectors or a bound on the vector infinity norm of $\bm{UV}'$, while \cite{rpca,rpca2,hsu2011robust} do.
Here $\bm{L} = [\bm{\ell}_{1},\dots,\bm{\ell}_{t_{\max}}] \overset{\mathrm{SVD}}= \bm{U\Sigma V}'$.
The reason for our stronger requirement on the product $sr$ is because we study an online algorithm, ReProCS, that recovers the sparse vector $\xt$ at each time $t$  rather than in a batch or a piecewise batch fashion.
Because of this, the sparse recovery step does not use the low dimensional structure of the new (and still unestimated) subspace.



Because we only require that the support changes after a given maximum allowed duration, it can be constant for a certain period of time (or move slowly as in Example \ref{motioneveryframe}). This is a substantially weaker assumption than the independent or uniformly random supports required by \cite{rpca} and \cite{mod_pcp}.
As we explain next, for $r>200$, this is also significantly weaker than the assumption used by \cite{hsu2011robust} (the result of \cite{hsu2011robust} is already an improvement of \cite{rpca2}). Corollary \ref{probcor} allows for at most $\frac{\alpha}{200}$ non-zero entries per row of  $[\bm{x}_{(k-1)\alpha},\dots,\bm{x}_{k\alpha-1}]$. However, in the next $\alpha$ frames, the same supports could be repeated. So if we consider the whole matrix $[\bm{x}_{1},\dots,\bm{x}_{t_{\max}}]$ then at most $\frac{t_{\max}}{200}$ non-zero entries per row are possible. This is an improvement over \cite{hsu2011robust} which requires at most $\frac{t_{\max}}{r}$ non-zero entries per row. In summary, {\em an important advantage} of  our result is that it allows for highly correlated support sets of $\xt$, which is important for applications such as video surveillance that involve one or more moving foreground objects or persons forming the sparse vector $\xt$.
See Section \ref{simulation} for simulations that illustrate this point.

In recent work, Feng et. al. \cite{rpca_stochatistic_optimization} propose a method for recursive robust PCA and prove a partial result for their algorithm.  The approach is to reformulate the PCP program and use this reformulation to develop a recursive algorithm that converges asymptotically to the solution of PCP (so the above comparisons still apply) as long as the basis estimate $\hat{\bm{P}}_t$ is full rank at each time $t$.
Since this result assumes something about the algorithm estimates, it is only a partial result.
Like our result, \cite{rpca_stochatistic_optimization} uses an initial estimate of the low dimensional subspace.  While we require knowledge of $c_{j,\rmnew}$ for each change, \cite{rpca_stochatistic_optimization} only needs the total rank $r$.

Another recent work that uses knowledge of the initial subspace estimate is modified PCP \cite{mod_pcp}.
The proof in \cite{mod_pcp} is a modification of the proof in \cite{rpca}.
Let $\bm{L} \overset{\mathrm{SVD}}= \bm{U\Sigma V'}$.
Modified PCP requires denseness of the columns of $\hat{\bm{P}}_{(0)}$ and $\bm{U}$ which if
$\operatorname{range}(\hat{\bm{P}}_{(0)}) \subseteq \operatorname{range}(\bm{L})$ is the same requirement as in \cite{rpca}.
Modified PCP also requires the same uniformly random supports as \cite{rpca}.
Where modified PCP improves the results of \cite{rpca} is in the assumptions on the right singular vectors $\bm{V}$ and vector infinity norm of $\bm{UV}'$.
These assumptions are only needed on the singular vectors of the unestimated part of the subspace.
Let $(\bm{I} - \hat{\bm{P}}_{(0)}\hat{\bm{P}}_{(0)}{}')\bm{L} \overset{\mathrm{SVD}} = \bm{U}_{\rmnew}\bm{\Sigma}_{\rmnew}{\bm{V}_{\rmnew}}'$.
Then modified PCP needs denseness of the columns of $\bm{V}_{\rmnew}$ and a bound on the vector infinity norm of $\bm{U}_{\rmnew}{\bm{V}_{\rmnew}}'$.
%
As with PCP, the main advantage of our result with respect to modified PCP is that we do not need uniformly random supports of the sparse vectors $\xt$.


Finally, our subspace change model only allows for adding new directions to the subspace. This is a valid model if, at time $t$, the goal is to estimate the column span of the matrix $\bm{L}_t:=[\bm{\ell}_1, \bm{\ell}_2, \dots, \bm{\ell}_t]$, which is the goal in robust PCA. However, when $\xt$ is the quantity of interest and $\lt$ is the large but structured noise, then this model can be restrictive (the denseness assumption imposes a bound on the rank $r$ see \eqref{srn}). In this case a better model would be one that also allows removal of directions from $\bm{P}_{(j)}$,  for example model 7.1 of \cite{ReProCS_IT} and to study the ReProCS-cPCA algorithm introduced there. This will significantly relax the required denseness assumption, and will be done in future work.



\section{Proof of Theorem \ref{thm1}}\label{pfsketch}

In this section we prove our main result.
To do this, we need new proof techniques that allow us to get a correctness result.
Except for the partial results in \cite{ReProCS_IT} and \cite{rpca_stochatistic_optimization}, all existing work on sparse plus low rank matrix decomposition is for batch methods.
We refer to these as partial results, because both require assumptions on estimates produced by the algorithm.
Like \cite{ReProCS_IT}, our proof cannot just be a combination of a sparse recovery result and a result for PCA, because in our PCA step,
the error between the estimated value of $\lt$ and its true value is correlated with $\lt$.
Almost all existing work on finite sample PCA assumes that the two are uncorrelated, e.g. \cite{nadler}.
Our proof is inspired by that of \cite{ReProCS_IT},
but we need a new approach to analyze the subspace estimate update step (step \ref{PCA} in Algorithm \ref{reprocs})
in order to remove the assumption on intermediate algorithm estimates used by the result of \cite{ReProCS_IT}.
The key new idea is to leverage the fact that, because of exact support recovery,
the error $\et := \hat{\bm{x}}_t - \xt = \lt - \hat{\bm{\ell}}_t$ is supported on $\mathcal{T}_t$.
Also, our support change model ensures that $\mathcal{T}_t$ changes at least every so often.
Together, this ensures that the matrix $\sum_t \et {\et} '$ can be written as a block-banded matrix with only three bands,
and the dominant term in  $\E[\sum_t (\I - \hat{\bm{P}}_{(j-1)}\hat{\bm{P}}_{(j-1)}{}') \lt {\et} ']$ can be written as a product of a full matrix and a block banded matrix with only three bands.  Here $\E[\cdot]$ denotes expected value conditioned on accurate recovery through the previous subspace update step.

Our proof strategy is to first show that given an accurate estimate of the current low-dimensional subspace, step 1 (sparse recovery step) of the algorithm will exactly recover the support, $\mathcal{T}_t$, and accurately recover $\xt$.  Since we know $\mt = \xt + \lt$, this also gives an accurate estimate of $\lt$.  The next, and difficult, step of the proof is to show that given the estimates $ \{ \hat{\bm{\ell}}_{t_j + (k-1)\alpha}, \dots, \hat{\bm{\ell}}_{t_j + k\alpha - 1} \}$ and the slow subspace change assumptions, step 2 (PCA step) of the algorithm will produce an accurate estimate of the newly added subspace.

\subsection{Three Main Lemmas}\label{keylemmas}

The proof of Theorem \ref{thm1} essentially follows from three main lemmas and one simple fact.  The actual proof of the theorem is in Appendix \ref{thmpf}.



\begin{definition}\label{defzetajk}
Define the following:
\begin{enumerate}
\item $\zeta_{j,*} := \|( \I - \hat{\bm{P}}_{(j-1)}\hat{\bm{P}}_{(j-1)}{}' ) \Pjm \|_2$
\item $ \zeta_{j,k} := \|(\bm{I} - \Phat_{(j-1)} \Phat_{(j-1)}{}' - \Phat_{(j),\rmnew,k} \Phat_{(j),\rmnew,k}{}') \bm{P}_{(j),\rmnew}\|_2 $
\item $\zeta_{j,*}^+ := (r_0 + (j-1)c)\zeta$

\item  $\zeta_{j,0}^+ := 1$,
\begin{align*}
 \zeta_{j,1}^+ & :=\frac{g\big( 8h^+ (\kappa_{s}^+)^2 (\phi^+)^2 + 2\kappa_{s}^+ \phi^+\big)}
{ D_{j,1} } 
+c\zeta \frac{\frac{(\zeta_{j,*}^+)^2}{c\zeta}f\big( 8h^+(\kappa_{s}^+)^2(\phi^+)^2 + {2\kappa_{s}^+\phi^+} + 2 \big) + \frac{5}{24}}
{D_{j,1} }
\end{align*}
where
\[
D_{j,1} := 1-(\zeta_{j,*}^+)^2(1+f) -\frac{c\zeta}{8} -  (\zeta_{j,*}^+)^2 f \big( 8h^+(\kappa_{s}^+)^2(\phi^+)^2 + {2\kappa_{s}^+\phi^+} + 2 \big) -  g \big( {8h^+(\kappa_{s}^+)^2} +{2\kappa_{s}^+\phi^+}\big) - \frac{5 c \zeta}{24}
\]

and for $k \geq 2$
\begin{align*}
 \zeta_{j,k}^+ & :=\zeta_{j,k-1}^+ \frac{g\big( 8h^+(\phi^+)^2\zeta_{j,k-1}^+ + 2\sqrt{8h^+}\phi^+\big)}
{ D_{j,k} } 
+c\zeta \frac{\frac{(\zeta_{j,*}^+)^2}{c\zeta}f\big( 8h^+(\phi^+)^2 + {2\sqrt{8h^+}\phi^+} + 2 \big) + \frac{5}{24}}
{D_{j,k} }
\end{align*}
where
\[
D_{j,k} := 1-(\zeta_{j,*}^+)^2(1+f) -\frac{c\zeta}{8} -  (\zeta_{j,*}^+)^2 f \big( 8h^+(\phi^+)^2 + {2\sqrt{8h^+}\phi^+} + 2 \big) - \zeta_{j,k-1}^+ g \big( {8h^+(\phi^+)^2\zeta_{j,k-1}^+} +{2\sqrt{8h^+}\phi^+}\big) - \frac{5 c \zeta}{24},
\]
$\kappa_{s}^+ := 0.0215$,
$h^+ := \frac{1}{200}$,
and $\phi^+ := 1.2$.
\end{enumerate}

\end{definition}
We use $\zeta_{j,*}$ as a measure of the error in estimating the previously existing subspace and $\zeta_{j,k}$ as a measure of the estimation error for the newly added subspace.  As will be shown, $\zeta_{j,*}^+$ and $\zeta_{j,k}^+$ are the high probability upper bounds on $\zeta_{j,*}$ and $\zeta_{j,k}$ respectively, under the conditions of Theorem \ref{thm1}.

The following 2 Lemmas are proved in the Appendix.

\begin{lem}[{Exponential decay of $\zeta_{j,k}^+$ (similar to \cite[Lemma 6.1]{ReProCS_IT})}] \label{expzeta}
Assume that the bounds on $\zeta$ from Theorem \ref{thm1} hold.
Then $\zeta_{j,0}^+ = 1$ and $\zeta_{j,k}^+ \leq 0.72^k + 0.83c\zeta$
for $k = 1, 2, \dots, K$.
Also, $\zeta_{j,k}^+ \leq 0.1$ for all $k \geq 1$.
\end{lem}


\begin{lem}[{Sparse Recovery Lemma (similar to \cite[Lemma 6.4]{ReProCS_IT}}]\label{cslem}
Assume that all of the conditions of Theorem \ref{thm1} hold.  Let
\[
\bm{b}_t = (\bm{I} - \Phat_{t-1} \Phat_{t-1}{}') \bm{\ell}_t
\]
 be the noise seen by the sparse recovery step.  If $\zeta_{j',*} \leq \big( r_0 + (j' - 1)c \big)\zeta$, $ j' = 1, \dots, j$ and $\zeta_{j,k'} \leq \zeta_{j,k'}^+$ for $k' = 1, \dots, k-1$, then for $t \in\mathcal{I}_{j,k}$,

\begin{enumerate}

\item  the support of $\xt$ is recovered exactly i.e. $\hat{\mathcal{T}}_t = \mathcal{T}_t$

\item $\phi_t := \| [ ({\bm{\Phi}_{t})_{\mathcal{T}_t}}'(\bm{\Phi}_{t})_{\mathcal{T}_t}]^{-1} \|_2 \leq \phi^+ := 1.2$.

\item $\et$ satisfies:
\begin{align}\label{etdef0}
\et := \hat{\bm{x}}_t - \xt = \bm{I}_{\mathcal{T}_t} {(\bm{\Phi}_{t})_{\mathcal{T}_t}}^{\dag} \bm{\bm{b}}_t = \bm{I}_{\mathcal{T}_t} [ ({\bm{\Phi}_{t})_{\mathcal{T}_t}}'(\bm{\Phi}_{t})_{\mathcal{T}_t}]^{-1}  {\bm{I}_{\mathcal{T}_t}}' \bm{\Phi}_{t} \bm{\ell}_t
\end{align}
and so
\begin{align*}
\|\bm{e}_t\|_2 &\leq
\begin{cases}
 \phi^+ (\zeta_{j,*}^+ \sqrt{r}\gamma +  \zeta_{j,k-1}^+ \sqrt{c}\gamma_{\rmnew}) & t \in [t_j,t_j +d]\\
\phi^+ (\zeta_{j,*}^+ \sqrt{r}\gamma +  c\zeta\sqrt{c}\gamma) & t \in (t_j +d,t_{j+1})
\end{cases}\\
 &\le
\begin{cases}
1.2 \left(1.83\sqrt{\zeta} +  (0.72)^{k-1}\sqrt{c}\gamma_{\rmnew}  \right) &  t \in [t_j,t_j +d] \\
1.2\left( 2\sqrt{\zeta} \right) & t \in (t_j + d, t_{j+1})
\end{cases}
\end{align*}

\end{enumerate}

\end{lem}

Notice from the proof that unlike \cite{ReProCS_IT}, Lemma \ref{cslem} does not use denseness of $(\bm{I} - \bm{P}_{(j),\rmnew}{\bm{P}_{(j),\rmnew}}')\hat{\bm{P}}_{(j),\rmnew,k}$.

The next lemma gives a lower bound on the probability of the newly added subspace being estimated accurately.
\begin{lem}[Subspace Recovery Lemma]\label{PCAlem}\label{mainlem}
Assume that all of the conditions of Theorem \ref{thm1} hold.  Then,
\begin{multline*}
\mathbb{P}\Big(  \zeta_{j,k} \leq \zeta_{j,k}^+ \ \big| \ \zeta_{j',*} \leq \zeta_{j',*}^+ \text { for } j' = 1, \dots, j
 \text{ and } \zeta_{j,k'} \leq \zeta_{j,k'}^+ \text{ for } k' = 1, \dots, k-1 \Big) \geq p(\alpha,\zeta)
\end{multline*}
where $p(\alpha,\zeta) = 1 - p_a(\alpha,\zeta) - p_b(\alpha,\zeta) - p_c(\alpha,\zeta)$ is a quantity that is increasing in $\alpha$.  The definitions of $p_a(\alpha,\zeta), p_b(\alpha,\zeta)$, and $p_c(\alpha,\zeta)$ can be found in Lemma \ref{termbnds}.
\end{lem}

Unlike \cite{ReProCS_IT}, the proof of this lemma does not use denseness of
$(\bm{I} - \hat{\bm{P}}_{(j),\rmnew}\hat{\bm{P}}_{(j),\rmnew}{}')\bm{P}_{(j),\rmnew,k}$.

The proof of Theorem \ref{thm1} follows by the above three lemmas and the following fact.

\begin{fact}\label{zetafact}
If $\zeta_{j,*} \leq \zeta_{j,*}^+$ and $\zeta_{j,K}\leq \zeta_{j,K}^+$, then $\zeta_{j+1,*} \leq \zeta_{j+1,*}^+$.
\end{fact}
\begin{proof}
Because $\hat{\bm{P}}_{j-1} \perp \hat{\bm{P}}_{(j),\rmnew,k}$, we have that $\zeta_{j+1,*} \leq \zeta_{j,*} + \zeta_{j,K}$.  The choice of $K$ and Lemma \ref{expzeta} imply that $\zeta_{j,K}^+ \leq c\zeta$, so $\zeta_{j+1,*} \leq \zeta_{j,*} + c\zeta$.  Finally, notice that $\zeta_{j+1,*}^+ = \zeta_{j,*}^+ + c\zeta$, so $\zeta_{j+1,*} \leq \zeta_{j+1,*}^+$ as desired.
\end{proof}

%

\subsection{A Lemma for Proving Lemma \ref{PCAlem} }

As stated previously, the key contribution of this work is to remove the assumption on algorithm estimates made by previous work.
We give the main lemma used to do this below.

\begin{lem}\label{blockdiag}
Consider a sequence of $s_t \times s_t$ ($s_t = |\mathcal{T}_t|$) symmetric positive-semidefinite matrices $\bm{A}_t$ such that
$\| \bm{A}_t\|_2 \leq \sigma^+$ for all $t$.
Let $\bm{M} = \sum_{t=a}^{b} \bm{I}_{\mathcal{T}_t} \bm{A}_t {\bm{I}_{\mathcal{T}_t}}'$ be an $n \times n$ matrix ($\I$ is an $n\times n$ identity matrix).
If the $\mathcal{T}_t$ satisfy \eqref{union} and $b - a \leq \beta$,
then
\[
\|\bm{M}\|_2 \leq 8 \sigma^+ h^*(\beta) .
\]
Notice that $\bm{M}$ is a large matrix formed by placing the small matrices $\bm{A}_t$ on their respective indices $\mathcal{T}_t$.
\end{lem}

Lemma \ref{blockdiag} is used in the proof of Lemma \ref{PCAlem} to bound the norm of $ \| \E [ \sum_{t} \et {\et}' ] \|_2$
and $ \| \E [ \sum_{t} (\I - \hat{\bm{P}}_{(j-1)}\hat{\bm{P}}_{(j-1)}{}')\bm{\ell}_t {\et}' ] \|_2$,
both of which have support structure governed by the support of $\et$, which by Lemma \ref{cslem},
is $\mathcal{T}_t$.  Here the expectation is conditioned on accurate recovery in the previous interval,
and the sum is taken over $\alpha$ time instants.

\begin{remark}
For the disjoint supports of Example \ref{disjointsupp}, the above Lemma can be modified to $\|\bm{M}\|_2 \leq 2\sigma^+ h^*(\beta)$, so in Theorem \ref{thm1} would would only need to assume $h^+ = \frac{1}{22}$.
\end{remark}

The key to proving Lemma \ref{blockdiag} is using the fact that
$\mathcal{T}_{t} \subseteq \mathcal{T}_{(i)} \cup \mathcal{T}_{(i+1)}$
to write $\bm{M}$ as a block tridiagonal matrix
with blocks corresponding to the $\mathcal{T}_{(i)}$.

\begin{proof}
First notice that because $b-a\leq\beta$, there exists a $u$ such that $(u-1)\beta \leq a \leq b \leq (u+1)\beta - 1$.  In other words, $[a,b]$ overlaps with at most two of the intervals $[(u-1)\beta,u\beta-1]$. So we can write $\bm{M}$ as
\[
\bm{M} = \sum_{t=a}^{u\beta-1} \bm{I}_{\mathcal{T}_t} \bm{A}_t {\bm{I}_{\mathcal{T}_t}}' + \sum_{t=u\beta}^{b} \bm{I}_{\mathcal{T}_t} \bm{A}_t {\bm{I}_{\mathcal{T}_t}}' .
\]
Let $\bm{M}_1 =  \sum_{t=a}^{u\beta-1} \bm{I}_{\mathcal{T}_t} \bm{A}_t {\bm{I}_{\mathcal{T}_t}}'$ and $\bm{M}_2 = \sum_{t=u\beta}^{b} \bm{I}_{\mathcal{T}_t} \bm{A}_t {\bm{I}_{\mathcal{T}_t}}'$.
We will show that $\|\bm{M}_1\|_2 \leq 4\sigma^+h^*(\beta)$. 

First notice that because the $\bm{A}_t$ are positive semi-definite,
\[
\bm{M}_1 = \sum_{t=a}^{u\beta-1} \bm{I}_{\mathcal{T}_t} \bm{A}_t {\bm{I}_{\mathcal{T}_t}}' \preceq 
\sum_{t=(u-1)\beta}^{u\beta-1} \bm{I}_{\mathcal{T}_t} \bm{A}_t {\bm{I}_{\mathcal{T}_t}}'
\]
So define $\tilde{\bm{M}}_1 := \sum_{t=(u-1)\beta}^{u\beta-1} \bm{I}_{\mathcal{T}_t} \bm{A}_t {\bm{I}_{\mathcal{T}_t}}'$.  Then a bound on $\|\tilde{\bm{M}}_1\|_2 $ gives a bound on $\|\bm{M}_1\|_2$.

Let $\mathcal{T}_{(i),u}$ ($i = 1,\dots,l_u$) be the optimizer in \eqref{hstar} (the definition of $h^*(\beta)$).  In the remainder of the proof we remove the subscript $u$ for convenience.  So $\mathcal{T}_{(i)}$ refers to $\mathcal{T}_{(i),u}$. Now consider a time $t$ for which $\mathcal{T}_t \subseteq \mathcal{T}_{(i)} \cup \mathcal{T}_{(i+1)}$.
\footnote{
If $\mathcal{T}_t \subseteq \mathcal{T}_{(i)}$,
then either $\mathcal{T}_{(i-1)} \cup \mathcal{T}_{(i)}$
or $\mathcal{T}_{(i)} \cup \mathcal{T}_{(i+1)}$
can be used to construct $\bm{A}_{t,\full}$.
The choice is inconsequential.
}
Define $\bm{A}_{t,\full}$ to be $\bm{A}_t$ with rows and columns of zeros appropriately inserted so that
\begin{equation}\label{full}
\bm{I}_{\mathcal{T}_t} \bm{A}_t {\bm{I}_{\mathcal{T}_t}}' = \bm{I}_{\mathcal{T}_{(i)} \cup \mathcal{T}_{(i+1)}} \bm{A}_{t,\full} {\bm{I}_{\mathcal{T}_{(i)}\cup \mathcal{T}_{(i+1)}}}'.
\end{equation}
Such an $\bm{A}_{t,\full}$ exists because $\mathcal{T}_t \subseteq \mathcal{T}_{(i)} \cup \mathcal{T}_{(i+1)} $.
Notice that
\begin{equation}\label{normequal}
\|\bm{A}_{t,\full}\|_2 = \|\bm{A}_{t}\|_2
\end{equation}
because $\bm{A}_{t,\full}$ is permutation similar to
\[
\left[
\begin{array}{cc}
\bm{A}_t & 0 \\
0    & 0
\end{array}
\right]
\]

Since $\mathcal{T}_{(i)}$ and $\mathcal{T}_{(i+1)}$ are disjoint, we can, after permutation similarity, correspondingly partition $\bm{A}_{t,\full}$ for the $t$ such that $\mathcal{T}_t \subseteq \mathcal{T}_{(i)} \cup \mathcal{T}_{(i+1)}$ as
\[
\bm{A}_{t,\full} =
\left[
\begin{array}{ccc}
\bm{A}_{t,\full}^{(i,i)} & \bm{A}_{t,\full}^{(i,i+1)} \\
\bm{A}_{t,\full}^{(i+1,i)} & \bm{A}_{t,\full}^{(i+1,i+1)}
\end{array}
\right].
\]

Notice that because $\bm{A}_t$ is symmetric, $\bm{A}_{t,\full}^{(i+1,i)} = \big( \bm{A}_{t,\full}^{(i,i+1)} \big)'$.

Then,
\begin{align*}
\tilde{\bm{M}}_1 &= \sum_{t=(u-1)\beta}^{u\beta-1} \bm{I}_{\mathcal{T}_t} \bm{A}_t {\bm{I}_{\mathcal{T}_t}}' \\
&= \sum_{i=1}^{l-1} \sum_{t:\mathcal{T}_t \subseteq \mathcal{T}_{(i)}\cup \mathcal{T}_{(i+1)}}  \bm{I}_{\mathcal{T}_{(i)} \cup \mathcal{T}_{(i+1)}} \bm{A}_{t,\full} {\bm{I}_{\mathcal{T}_{(i)}\cup \mathcal{T}_{(i+1)}}}' \qquad \text{by \eqref{full} } \\
&= \sum_{i=1}^{l-1}\sum_{t:\mathcal{T}_t \subseteq \mathcal{T}_{(i)}\cup \mathcal{T}_{(i+1)}}  [ \bm{I}_{\mathcal{T}_{(i)}} \ \bm{I}_{\mathcal{T}_{(i+1)}}] \bm{A}_{t,\full} \left[\begin{array}{c} {\bm{I}_{\mathcal{T}_{(i)}}}' \\ {\bm{I}_{\mathcal{T}_{(i+1)}}}' \end{array}\right] \\
&= \sum_{i=1}^{l-1} \sum_{t:\mathcal{T}_t \subseteq \mathcal{T}_{(i)}\cup \mathcal{T}_{(i+1)}}  \bm{I}_{\mathcal{T}_{(i)}}\bm{A}_{t,\full}^{(i,i)}{\bm{I}_{\mathcal{T}_{(i)}}}' + \bm{I}_{\mathcal{T}_{(i)}}\bm{A}_{t,\full}^{(i,i+1)}{\bm{I}_{\mathcal{T}_{(i+1)}}}' + \bm{I}_{\mathcal{T}_{(i+1)}}\bm{A}_{t,\full}^{(i+1,i)}{\bm{I}_{\mathcal{T}_{(i)}}}' + \bm{I}_{\mathcal{T}_{(i+1)}} \bm{A}_{t,\full}^{(i+1,i+1)} {\bm{I}_{\mathcal{T}_{(i+1)}}}' \\
&= \sum_{i=1}^{l-1} \left[ \bm{I}_{\mathcal{T}_{(i)}}\left( \sum_{t:\mathcal{T}_t \subseteq T_{(i-1)}\cup \mathcal{T}_{(i)}} \bm{A}_{t,\full}^{(i,i)} + \sum_{t:\mathcal{T}_t \subseteq \mathcal{T}_{(i)}\cup \mathcal{T}_{(i+1)}} \bm{A}_{t,\full}^{(i,i)} \right) {\bm{I}_{\mathcal{T}_{(i)}}}'\right.  \\
& \hspace{1 in}\left. + \bm{I}_{\mathcal{T}_{(i)}} \left(\sum_{t:\mathcal{T}_t \subseteq \mathcal{T}_{(i)}\cup \mathcal{T}_{(i+1)}} \bm{A}_{t,\full}^{(i,i+1)}\right) {\bm{I}_{\mathcal{T}_{(i+1)}}}'  +  \bm{I}_{\mathcal{T}_{(i+1)}} \left(\sum_{t:\mathcal{T}_t \subseteq \mathcal{T}_{(i)}\cup \mathcal{T}_{(i+1)}} \bm{A}_{t,\full}^{(i+1,i)}\right) {\bm{I}_{\mathcal{T}_{(i)}}}' \right]
\end{align*}

Because $\mathcal{T}_{(i)}$ and $T_{(j)}$ are disjoint for $i\neq j$, $\tilde{\bm{M}}_1$ has a block tridiagonal structure (by a permutation similarity if necessary):
\begin{equation}\label{blockstruct}
\tilde{\bm{M}}_1 =
\left[
\begin{array}{cccc}
\bm{B}_{(1)} & \bm{C}_{(1)} & 0 & 0\\
\bm{C}_{(1)}' & \bm{B}_{(2)} & \ddots & 0 \\
0          & \ddots &     \ddots & \bm{C}_{(l-1)} \\
0 & 0 & \bm{C}_{(l-1)}' & \bm{B}_{(l)}
\end{array}
\right]
\end{equation}

where
\begin{equation}\label{doubledose}
\bm{B}_{(i)} = \sum_{t:\mathcal{T}_t \subseteq T_{(i-1)}\cup \mathcal{T}_{(i)}} \bm{A}_{t,\full}^{(i,i)} + \sum_{t:\mathcal{T}_t \subseteq \mathcal{T}_{(i)}\cup \mathcal{T}_{(i+1)}} \bm{A}_{t,\full}^{(i,i)}
\end{equation}
and
\begin{equation}\label{B(i)}
\bm{C}_{(i)} = \sum_{t:\mathcal{T}_t \subseteq \mathcal{T}_{(i)}\cup \mathcal{T}_{(i+1)}} \bm{A}_{t,\full}^{(i,i+1)}.
\end{equation}

Now we proceed to bound $\|\tilde{\bm{M}}_1\|_2$:
\begin{align*}
\|\tilde{\bm{M}}_1\|_2 & =
\left\| \begin{array}{cccc}
\bm{B}_{(1)} & \bm{C}_{(1)}  & 0 &0 \\
{\bm{C}_{(1)}}' & \ddots & \ddots & 0 \\
0     &  \ddots & \ddots & \bm{C}_{(l-1)} \\
0 & 0 & {\bm{C}_{(l-1)}}' & \bm{B}_{(l)}
\end{array}\right\|_2 \nonumber\\
&\leq
\left\| \begin{array}{cccc}
\bm{B}_{(1)} & 0  & 0 &0 \\
0 & \ddots & 0 & 0 \\
0     &  0 & \ddots &0 \\
0 & 0 & 0 & \bm{B}_{(l)}
\end{array}\right\|_2
+
\left\| \begin{array}{cccc}
0 & \bm{C}_{(1)}  & 0 &0 \\
0 & 0 & \ddots & 0 \\
0     &  0 & 0 & \bm{C}_{(l-1)} \\
0 & 0 & 0 & 0
\end{array}\right\|_2
+
\left\| \begin{array}{cccc}
0 &0  & 0 &0 \\
{\bm{C}_{(1)}}' & 0 & 0 & 0 \\
0     &  \ddots & 0 &0 \\
0 & 0 & {\bm{C}_{(l-1)}}' & 0
\end{array}\right\|_2. \\
\end{align*}

Call the middle matrix $\bm{C}$, and observe that $\bm{CC}'$ is block diagonal with blocks $\bm{C}_{(i)}{\bm{C}_{(i)}}'$.  So $\|\bm{C}\|_2 = \max_{i}\|\bm{C}_{(i)}\|_2$.

Therefore,
\begin{align*}
\|\tilde{\bm{M}}_1\|_2 &\leq \max_{i}\|\bm{B}_{(i)}\|_2 + 2\max_{i}\|\bm{C}_{(i)}\|_2 \\
&= \max_{i} \bigg\| \sum_{t:\mathcal{T}_t \subseteq T_{(i-1)}\cup \mathcal{T}_{(i)}} \bm{A}_{t,\full}^{(i,i)} + \sum_{t:\mathcal{T}_t \subseteq \mathcal{T}_{(i)}\cup \mathcal{T}_{(i+1)}} \bm{A}_{t,\full}^{(i,i)}\bigg\|_2 + 2 \max_{i} \bigg\| \sum_{t:\mathcal{T}_t \subseteq T_{(i-1)}\cup\T_{(i)}} \bm{A}_{t,\full}^{(i,i+1)} \bigg\|_2  \quad \text{ by \eqref{doubledose} and \eqref{B(i)}}\nn\\
&\leq \max_{i} \left( \sum_{t:\mathcal{T}_t \subseteq T_{(i-1)}\cup \mathcal{T}_{(i)}}\big\|  \bm{A}_{t,\full}^{(i,i)}\big\|_2 + \sum_{t:\mathcal{T}_t \subseteq \mathcal{T}_{(i)}\cup \mathcal{T}_{(i+1)}}\big\|  \bm{A}_{t,\full}^{(i,i)}\big\|_2 \right) + 2 \max_{i} \sum_{t:\mathcal{T}_t \subseteq \mathcal{T}_{(i)}\cup \mathcal{T}_{(i+1)}}\big\|  \bm{A}_{t,\full}^{(i,i+1)} \big\|_2 \nn \\
%
%
&\leq \max_{i} \left( \sum_{t:\mathcal{T}_t \subseteq T_{(i-1)}\cup \mathcal{T}_{(i)}}\big\|  \bm{A}_{t}\big\|_2 + \sum_{t:\mathcal{T}_t \subseteq T_{(i)}\cup\T_{(i+1)}}\big\|  \bm{A}_{t}\big\|_2 \right) + 2 \max_{i} \sum_{t:\mathcal{T}_t \subseteq T_{(i)}\cup\T_{(i+1)}}\big\|  \bm{A}_{t} \big\|_2 \nn \qquad \text{ by \eqref{normequal}}\\
&\leq \max_{i} \left( \sum_{t:\mathcal{T}_t \subseteq T_{(i-1)}\cup \mathcal{T}_{(i)}}\sigma^{+} + \sum_{t:\mathcal{T}_t \subseteq \mathcal{T}_{(i)}\cup \mathcal{T}_{(i+1)}}\sigma^{+} \right) + 2 \max_{i} \sum_{t:\mathcal{T}_t \subseteq \mathcal{T}_{(i)}\cup \mathcal{T}_{(i+1)}}\sigma^{+} \nn \\
&\leq \sigma^{+}h^*(\beta) + \sigma^{+}h^*(\beta) + 2 \sigma^{+}h^*(\beta) \qquad \text{ by the definition of $h^*(\beta)$ and optimality of $\mathcal{T}_{(i)}$ }\nn \\
&\leq {4}{\sigma^+}{h^*(\beta)} \nn
\end{align*}
The exact same argument shows $\|\bm{M}_2\|_2 \leq {4}{\sigma^+}{h^*(\beta)}$, and so by the triangle inequality we have
\[
\|\bm{M}\|_2 \leq 8 \sigma^+ h^*(\beta) .
\]
\end{proof}

\subsection{Proof of Lemma \ref{PCAlem} }

To prove Lemma \ref{PCAlem} we use:
\begin{enumerate}
\item the $\sin \theta$ theorem of Davis and Kahan \cite{davis_kahan},
\item the expression for $\et$ from Lemma \ref{cslem},
\item Lemma \ref{blockdiag} that bounds the norm of a block banded matrix, and
\item the matrix Hoeffding bounds from \cite{tail_bound}.
\end{enumerate}

\begin{remark}
Because this lemma applies for all $j = 1, \dots, J$, we remove the subscript $j$ for
convenience.  So $\zeta_{k}^+$ refers to $\zeta_{j,k}^+$,
$\Phat_{\rmnew,k} $ refers to $\Phat_{(j),\rmnew,k}$, etc.
Also, $\bm{P}_{*}$ refers to $\bm{P}_{(j-1)}$ and $\hat{\bm{P}}_{*}$ refers to $\hat{\bm{P}}_{(j-1)}$.
\end{remark}

The proof of Lemma \ref{PCAlem} requires several definitions.

\begin{definition}
\[
\mathcal{I}_{j,k} := [t_j + (k-1)\alpha, t_j + k\alpha -1]
\]
for $j = 1,\dots,J$ and $k = 1,\dots,K,$
and
\[
\mathcal{I}_{j,K+1} := [t_j +K\alpha, t_{j+1} - 1]
\]
\end{definition}

\begin{definition}\label{defn_Phi}

Define the following
\ben
\item $\Phat_{\rmnew,0} = [.]$ (empty matrix)
\item $\bm{\Phi}_{(k)} := \bm{I} - \Phat_{*} {\Phat_{*}}{}' - \Phat_{\rmnew,k} \Phat_{\rmnew,k}{}' $.

Because $\Phat_{*}\perp\Phat_{\rmnew,k}$, $\bm{\Phi}_{(k)} = ( \bm{I} - \Phat_{*} {\Phat_{*}}{}' )( \bm{I} - \Phat_{\rmnew,k} \Phat_{\rmnew,k}{}' )$

Notice (from Algorithm \ref{reprocs}) that for $t\in \mathcal{I}_{j,k}$, $\bm{\Phi}_t = \bm{\Phi}_{(k-1)}$.

\item $\bm{D}_{\rmnew,k}$, $\bm{D}_{\rmnew}$, $\bm{D}_{*,k}$ and $\bm{D}_{*}$
\ben
\item $\bm{D}_{\rmnew,k} := \bm{\Phi}_{(k)} \bm{P}_{\rmnew}$ and
$\bm{D}_{\rmnew} := \bm{D}_{\rmnew,0} = \bm{\Phi}_{(0)} \bm{P}_{\rmnew}$.

\item $\bm{D}_{*,k} := \bm{\Phi}_{(k)} \bm{P}_{*}$ and $\bm{D}_{*} := \bm{D}_{*,0} = \bm{\Phi}_{(0)} \bm{P}_{*}$.

\item Notice that $\zeta_{0} = \|\bm{D}_{\rmnew}\|_2$, $\zeta_{k} = \|\bm{D}_{\rmnew,k}\|_2$, $\zeta_{*} = \|\bm{D}_{*}\|_2$.
Also, clearly, $\|\bm{D}_{*,k}\|_2 \le \zeta_{*}$.
\een
\een
\end{definition}

\begin{definition}\label{defHk}

\begin{enumerate}
\item Let $\bm{D}_{\rmnew} \overset{QR}{=} \bm{E}_{\rmnew} \bm{R}_{\rmnew}$ denote its reduced QR decomposition, i.e. let $\bm{E}_{\rmnew}$ be a basis matrix for $\Span(\bm{D}_{\rmnew})$ and let $\bm{R}_{\rmnew} = {\bm{E}_{\rmnew}}'\bm{D}_{\rmnew}$.

\item Let $\bm{E}_{\rmnew,\perp}$ be a basis matrix for the orthogonal complement of $\Span(\bm{E}_{\rmnew})=\Span(\bm{D}_{\rmnew})$. To be precise, $\bm{E}_{\rmnew,\perp}$ is a $n \times (n-c_{\rmnew})$ basis matrix that satisfies ${\bm{E}_{\rmnew,\perp}}'\bm{E}_{\rmnew}=\bm{0}$.

\item Using $\bm{E}_{\rmnew}$ and $\bm{E}_{\rmnew,\perp}$, define $\bm{A}_{k}$ and $\bm{A}_{k,\perp}$ as
\begin{align*}
\bm{A}_{k} &:= \frac{1}{\alpha} \sum_{t \in \mathcal{I}_{j,k}} {\bm{E}_{\rmnew}}' \bm{\Phi}_{(0)} \bm{\ell}_t {\bm{\ell}_t}' \bm{\Phi}_{(0)} \bm{E}_{\rmnew} \\
\bm{A}_{k,\perp} &:= \frac{1}{\alpha} \sum_{t \in \mathcal{I}_{j,k}} {\bm{E}_{\rmnew,\perp}}' \bm{\Phi}_{(0)} \bm{\ell}_t {\bm{\ell}_t}' \bm{\Phi}_{(0)} \bm{E}_{\rmnew,\perp}
\end{align*}
and let
\[
\bm{\mathcal{A}}_{k} := \left[ \begin{array}{cc} \bm{E}_{\rmnew} & \bm{E}_{\rmnew,\perp} \\ \end{array} \right]
\left[\begin{array}{cc} \bm{A}_{k} \ & 0 \ \\ 0 \ & \bm{A}_{k,\perp}  \\ \end{array} \right]
\left[ \begin{array}{c} {\bm{E}_{\rmnew}}' \\ {\bm{E}_{\rmnew,\perp}}' \\ \end{array} \right]
\]

\item Define $\bm{\mathcal{H}}_{k}$ so that
\begin{align*}
\bm{\mathcal{A}}_{k} + \bm{\mathcal{H}}_{k} = \frac{1}{\alpha} \sum_{t \in \mathcal{I}_{j,k}} \bm{\Phi}_{(0)} \hat{\bm{\ell}}_t \hat{\bm{\ell}}_t{}' \bm{\Phi}_{(0)}
\end{align*}
is the matrix whose top $c_{j,\rmnew}$ singular vectors form $\Phat_{\rmnew,k}$ (see step \ref{PCA} of Algorithm \ref{reprocs}).  So $\bm{\mathcal{A}}_{k} + \bm{\mathcal{H}}_{k}$ has eigendecomposition
\begin{align*}
\bm{\mathcal{A}}_{k} + \bm{\mathcal{H}}_{k} \overset{\mathrm{EVD}}{=} \left[ \begin{array}{cc} \Phat_{\rmnew,k} & \Phat_{\rmnew,k,\perp} \\ \end{array} \right]
\left[\begin{array}{cc} \bm{\Lambda}_k \ & 0 \ \\ 0 \ & \ \bm{\Lambda}_{k,\perp} \\ \end{array} \right]
\left[ \begin{array}{c} \Phat_{\rmnew,k}{}' \\ \Phat_{\rmnew,k,\perp}{}' \\ \end{array} \right].
\end{align*}

\end{enumerate}

\end{definition}

\begin{lem}[$\sin\theta$ theorem \cite{davis_kahan} using the above notation]
If $\lambda_{\min}(\bm{A}_k) > \lambda_{\max}(\bm{\Lambda}_{k,\perp})$, then
\[
\| (\I - \Phat_{\rmnew,k} \Phat_{\rmnew,k}{}')\bm{E}_{\rmnew}\|_2 \leq
\frac{\|\bm{\mathcal{H}}_k\bm{E}_{\rmnew}\|_2}{\lambda_{\min}(\bm{A}_k) - \lambda_{\max}(\bm{\Lambda}_{k,\perp})}
\]
\end{lem}

The next lemma follows from the $\sin\theta$ lemma and Weyl's theorem.  It is taken from \cite{ReProCS_IT}.

\begin{lem}\label{zetakbnd}
If $\lambda_{\min}(\bm{A}_k) - \|\bm{A}_{k,\perp}\|_2 - \|\bm{\mathcal{H}}_k\|_2 >0$, then
\beq \label{zetakbound}
\zeta_k \leq  \frac{\|\bm{\mathcal{H}}_k\|_2}{\lambda_{\min} (\bm{A}_k) - \|\bm{A}_{k,\perp}\|_2 - \|\bm{\mathcal{H}}_k\|_2}
\eeq
\end{lem}

\begin{definition}\label{Gamma_k}
Define the random variable
\[
X_{k-1} := [\bm{a}_1 , \dots , \bm{a}_{t_j + (k-1)\alpha - 1}]
\]
and the set
\[
\Gamma_{k-1} = \{ X_{k-1} : \zeta_{j',*} \leq \zeta_{j',*}^+ \text{ for } j' = 1, \dots, j \text{ and } \zeta_{j,k'} \leq \zeta_{j,k'}^+ \text{ for } k' = 1, \dots, k -1  \}
\]
and let
$\Gamma_{k-1}^e$ denote the event $X_{k-1} \in \Gamma_{k-1}$.

To prove Corollary \ref{probcor} (or if otherwise considering random supports), also include $\{\T_{1},\T_{2}, \dots, \T_{t_{\max}}\}$ in the definition of $X_{k-1}$.  Because the supports are independent of the $\at$'s, when conditioned on $X_{k-1}$ they can be treated as constant.
\end{definition}

\begin{lem} \label{Dnew0_lem}
Define $\kappa_{s}^+ := 0.0215$.
Assume that the assumptions of Theorem \ref{thm1} hold.
Conditioned on $X_{k-1} \in \Gamma_{k-1}$,
\[
\| {\I_{\mathcal{T}}}' \bm{D}_{\new} \|_2 \leq \kappa_{s}^+
\]
for all $\mathcal{T}$ such that $|\mathcal{T}|\leq s$.
\end{lem}
\begin{proof}
Observe that $\bm{D}_{\rmnew,0} = (\bm{I} - \Phat_{j-1}\Phat_{j-1}{}') \bm{P}_\rmnew$.
Then $\|{\I_{\mathcal{T}}}' \bm{D}_{\rmnew} \|_2 = \|{\I_{\mathcal{T}}}'  (\bm{I} - \Phat_{j-1}\Phat_{j-1}{}') \bm{P}_\rmnew \|_2 \leq \| {\I_{\mathcal{T}}}' \bm{P}_\rmnew \|_2 + \|\Phat_{j-1}{}' \bm{P}_\rmnew \|_2 \leq {\kappa_s(\bm{P}_{\rmnew}) +  \zeta_*}$.
The event $X_{k-1} \in \Gamma_{k-1}$ implies that $\zeta_* \le \zeta_*^+ \le 0.0015$. Thus, the lemma follows.
\end{proof}

\begin{lem}[High probability bounds for each of the terms in the $\zeta_k$ bound (\ref{zetakbound})]\label{termbnds}
Assume the conditions of Theorem \ref{thm1} hold.  Also assume that $\mathbb{P}(X_{k-1} \in \Gamma_{k-1})>0$. Then, for $ k = 1,\dots,K$
\begin{enumerate}

\item $\mathbb{P} \left(\lambda_{\min} (\bm{A}_{k}) \geq  \lambda_{\rmnew}^- \left(1 -(\zeta_{j,*}^+)^2  - \frac{c \zeta}{12}\right) \big|X_{k-1}\in\Gamma_{k-1}\right) > 1- p_{a}(\alpha,\zeta)$
where
\[
p_{a}(\alpha,\zeta) := c \exp \left(\frac{-\alpha \zeta^2 (\lambda^-)^2}{8 \cdot 24^2 \cdot {\gamma_{\rmnew}}^4  } \right) + c \exp \left( \frac{-\alpha c^2 \zeta^2(\lambda^-)^2} {8 \cdot 24^2 \cdot 4^2}\right)
\]

\item $\mathbb{P}\left(\lambda_{\max}(\bm{A}_{k,\perp}) \leq \lambda_{\rmnew}^- \left( (\zeta_{j,*}^+)^2 f + \frac{c \zeta}{24}\right) \big| X_{k-1}\in\Gamma_{k-1} \right) > 1- p_b(\alpha,\zeta)$
where
\[
p_b (\alpha,\zeta) := (n-c) \exp \left(\frac{-\alpha c^2 \zeta^2 (\lambda^-)^2}{8 \cdot 24^2}\right)
\]

\item
$\mathbb{P}\left(\|\bm{\mathcal{H}}_{k}\|_2 \leq  b_k + \frac{5c\zeta\lambda^-}{24}) \ \big|X_{k-1}\in\Gamma_{k-1}\right) \geq 1 - p_c(\alpha,\zeta)$
where $b_k = b_2 + 2b_4 + 2b_6$, and
\begin{align}
b_{2,k} &= 
\begin{cases} 
8h^+(\phi^+)^2 \Big((\zeta_*^+)^2\lambda^+ + (\kappa_{s}^+)^2\lambda_{\rmnew}^+\Big) & k=1 \\
8h^+(\phi^+)^2 \Big((\zeta_*^+)^2\lambda^+ + (\zeta_{k-1}^+)^2\lambda_{\rmnew}^+\Big) & k\geq2 
\end{cases} \nn\\
\label{b4}
b_{4,k} &= 
\begin{cases} 
(\zeta_{*}^+)^2\lambda^+ + \kappa_{s}^+\lambda_{\new}^+ & k=1 \\
\left( (\zeta_*^+)^2 \lambda^+ + \zeta_{k-1}^+\lambda_{\rmnew}^+  \right)\left(\sqrt{8h^+}\phi^+\right) & k\geq2.
\end{cases}  \\
b_6 &= (\zeta_*^+)^2 \lambda^+  \nn
\end{align}

Also,
\begin{align*}
p_c(\alpha,\zeta) :=  & \ n \exp\left(\frac{-\alpha c^2\zeta^2 (\lambda^-)^2}{8 \cdot 24^2 \Big( \phi^+ (\sqrt{\zeta}+\sqrt{c}\gamma_{\rmnew}) \Big)^4}\right)
+  n \exp\left(\frac{-\alpha c^2\zeta^2 (\lambda^-)^2}{8\cdot 24^2 (\phi^+)^2 (\sqrt{\zeta}+\sqrt{c}\gamma_{\rmnew})^4}\right) \\
&+  n \exp\left(\frac{-\alpha c^2 \zeta^2 (\lambda^-)^2 }{8 \cdot 24^2 ( \zeta r^2 \gamma^2)^2}\right). \\
\end{align*}


\end{enumerate}

\end{lem}

\begin{proof}[Proof of Lemma \ref{PCAlem}]
Lemma \ref{PCAlem} now follows by combining Lemmas \ref{zetakbnd} and \ref{termbnds},
using $\lambda_{\rmnew}^- \geq \lambda^-$, $f = \frac{\lambda^+}{\lambda^-}$, $g = \frac{\lambda_{\rmnew}^+}{\lambda_{\rmnew}^-}$,
and defining
\begin{equation}
p(\alpha,\zeta) := 1 - p_{a}(\alpha,\zeta) - p_b(\alpha,\zeta) - p_c(\alpha,\zeta). \label{pk}
\end{equation}
\end{proof}

\begin{proof}[Proof of Lemma \ref{termbnds}]

The proof of the first two claims is similar to \cite{ReProCS_IT}, while the proof of the third is very different.

For convenience, we will use $\frac{1}{\alpha}\sum_t$ to denote $\frac{1}{\alpha} \sum_{t \in \mathcal{I}_{j,k}}$.

First observe that the matrices $\bm{D}_{\rmnew}$, $\bm{R}_{\rmnew}$, $\bm{E}_{\rmnew}$, $\bm{D}_{*}, \bm{D}_{\rmnew,k-1}$, $\bm{\Phi}_{(k-1)}$ are all functions of the random variable $X_{k-1}$.  Since $X_{k-1}$ is independent of any $a_{t}$ for $t \in  \mathcal{I}_{j,k}$, the same is true for the matrices  $\bm{D}_{\rmnew}$, $\bm{R}_{\rmnew}$, $\bm{E}_{\rmnew}$, $\bm{D}_{*}, \bm{D}_{\rmnew,k-1}$, $\bm{\Phi}_{(k-1)}$.

    All terms that we bound for the first two claims of the lemma are of the form $\frac{1}{\alpha} \sum_{t \in \mathcal{I}_{j,k}} \bm{Z}_t$ where $\bm{Z}_t= f_1(X_{k-1}) \bm{Y}_t f_2(X_{k-1})$, $\bm{Y}_t$ is a sub-matrix of $\bm{a}_t {\bm{a}_t}'$, and $f_1(.)$ and $f_2(.)$ are functions of $X_{k-1}$. Thus, conditioned on  $X_{k-1}$, the $\bm{Z}_t$'s are mutually independent. \label{X_at_indep} \label{zt_indep}

    All the terms that we bound for the third claim contain $\bm{e}_t$. Using Lemma \ref{cslem}, conditioned on $X_{k-1}$, $\et$ satisfies (\ref{etdef0}) with probability one whenever $X_{k-1} \in \Gamma_{k-1}$. Using (\ref{etdef0}), it is easy to see that all these terms are also of the above form whenever $X_{k-1} \in \Gamma_{k-1}$.  Thus, conditioned on $X_{k-1}$, the $\bm{Z}_t$'s for all the above terms are mutually independent, whenever $X_{k-1} \in \Gamma_{k-1}$.

The following are corollaries of the matrix Hoeffding inequality in \cite{tail_bound} and are proved in \cite{ReProCS_IT}.
\begin{corollary}[Matrix Hoeffding conditioned on another random variable for a nonzero mean Hermitian matrix \cite{tail_bound,ReProCS_IT}]\label{hoeffding_nonzero}
Given an $\alpha$-length sequence $\{\bm{Z}_t\}$ of random Hermitian matrices of size $n\times n$, a r.v. $X$, and a set ${\cal C}$ of values that $X$ can take. Assume that, for all $X \in \calc$, (i) $\bm{Z}_t$'s are conditionally independent given $X$; (ii) $\mathbf{P}(b_1 \bm{I} \preceq \bm{Z}_t \preceq b_2 \bm{I} | X) = 1$ and (iii) $b_3 \bm{I} \preceq \frac{1}{\alpha}\sum_t \E(\bm{Z}_t | X) \preceq b_4 \bm{I} $. Then for all $\epsilon > 0$,
\begin{align*}
\mathbb{P} \left( \lambda_{\max}\left(\frac{1}{\alpha}\sum_t \bm{Z}_t \right) \leq b_4 + \epsilon \Big | X \right)
\geq 1- n \exp\left(\frac{-\alpha \epsilon^2}{8(b_2-b_1)^2}\right) \ \text{for all} \ X \in \calc
\end{align*}
\begin{align*}
\mathbb{P} \left(\lambda_{\min}\left(\frac{1}{\alpha}\sum_t \bm{Z}_t \right) \geq b_3 -\epsilon \Big| X \right)
\geq  1- n \exp\left(\frac{-\alpha \epsilon^2}{8(b_2-b_1)^2} \right)  \text{for all} \ X \in \calc
\end{align*}
\end{corollary}

\begin{corollary}[Matrix Hoeffding conditioned on another random variable for an arbitrary nonzero mean matrix]\label{hoeffding_rec}
Given an $\alpha$-length sequence $\{\bm{Z}_t\}$ of random matrices of size $n\times n$, a r.v. $X$, and a set ${\mathcal{C}}$ of values that $X$ can take. Assume that, for all $X \in \calc$, (i) $\bm{Z}_t$'s are conditionally independent given $X$; (ii) $\mathbb{P}(\|\bm{Z}_t\|_2 \le b_1|X) = 1$ and (iii) $\|\frac{1}{\alpha}\sum_t \E( \bm{Z}_t|X)\|_2 \le b_2$. Then, for all $\epsilon >0$,
\begin{align*}
\mathbb{P} \left(\bigg\|\frac{1}{\alpha}\sum_t \bm{Z}_t \bigg\|_2 \leq b_2 + \epsilon \Big| X \right) 
\geq 1-(n_1+n_2) \exp\left(\frac{-\alpha \epsilon^2}{32 {b_1}^2}\right)  \ \text{for all} \ X \in \calc
\end{align*}
\end{corollary}
We will use the following fact applied to the above corollaries.
\begin{fact}\label{event}
For an event $\mathcal{E}$ and random variable $X$, $\mathbb{P}(\mathcal{E}|X) \geq p$ for all $X \in \mathcal{C}$ implies that
$\mathbb{P}(\mathcal{E}|X\in \mathcal{C}) \geq p$.
\end{fact}

To begin bounding the terms in \eqref{zetakbound}, first consider $\bm{A}_k := \frac{1}{\alpha} \sum_t {\bm{E}_{\rmnew}}' \bm{\Phi}_{(0)} \bm{\ell}_t {\bm{\ell}_t}' \bm{\Phi}_{(0)} \bm{E}_{\rmnew}$. Notice that ${\bm{E}_{\rmnew}}' \bm{\Phi}_{(0)} \bm{\ell}_t = \bm{R}_{\rmnew} \bm{a}_{t,\rmnew} + {\bm{E}_{\rmnew}}' \bm{D}_* \bm{a}_{t,*}$. Let $\bm{Z}_t = \bm{R}_{\rmnew} \bm{a}_{t,\rmnew} {\bm{a}_{t,\rmnew}}' {\bm{R}_{\rmnew}}'$ and let $\bm{Y}_t = \bm{R}_{\rmnew} \bm{a}_{t,\rmnew}{\bm{a}_{t,*}}' {\bm{D}_*}' {\bm{E}_{\rmnew}} +  {\bm{E}_{\rmnew}}' \bm{D}_* \bm{a}_{t,*}{\bm{a}_{t,\rmnew}}' {\bm{R}_{\rmnew}}'$, then
\beq
\bm{A}_k  \succeq \frac{1}{\alpha} \sum_t \bm{Z}_t + \frac{1}{\alpha} \sum_t \bm{Y}_t \label{lemmabound_1}
\eeq

Consider $\sum_t \bm{Z}_t = \sum_t \bm{R}_{\rmnew} \bm{a}_{t,\rmnew} {\bm{a}_{t,\rmnew}}' {\bm{R}_{\rmnew}}'$.
\begin{enumerate}
\item The $\bm{Z}_t$'s are conditionally independent given $X_{k-1}$.

\item Using a theorem of Ostrowoski \cite[Theorem 4.5.9]{hornjohnson},
 for all $X_{k-1} \in \Gamma_{k-1}$, $\lambda_{\min}\left( \E(\frac{1}{\alpha}\sum_t \bm{Z}_t|X_{k-1})\right) = \lambda_{\min}\left( \bm{R}_{\rmnew} \frac{1}{\alpha}\sum_t \E(\bm{a}_{t,\rmnew}{\bm{a}_{t,\rmnew}}') {\bm{R}_{\rmnew}}'\right) \ge \lambda_{\min} \left(\bm{R}_{\rmnew} {\bm{R}_{\rmnew}}'\right)\lambda_{\min} \left(\frac{1}{\alpha}\sum_t \E(\bm{a}_{t,\rmnew}{\bm{a}_{t,\rmnew}}')\right) \geq (1-(\zeta_{j,*}^+)^2)\lambda_{\rmnew}^-$.

\item Finally, conditioned on $X_{k-1}$,  $0 \preceq \bm{Z}_t  \preceq c {\gamma_{\rmnew}}^2 \bm{I}$ holds for all $X_{k-1} \in \Gamma_{k-1}$.
\end{enumerate}

Thus, applying Corollary \ref{hoeffding_nonzero} with $\epsilon = \frac{c\zeta\lambda^-}{24}$ and Fact \ref{event}, we get
\begin{equation}
\mathbb{P}\left(\lambda_{\min} \left(\frac{1}{\alpha} \sum_t \bm{Z}_t\right)
\geq   (1-(\zeta_*^+)^2)\lambda_{\rmnew}^-  - \frac{c\zeta\lambda^-}{24} \bigg| X_{k-1}\in\Gamma_{k-1}\right)
\geq  1- c \exp \left(\frac{-\alpha \zeta^2 (\lambda^-)^2}{8 \cdot 24^2 \cdot {\gamma_{\rmnew}}^4}\right).
 \label{lemma_add_A1}
\end{equation}

Consider $\bm{Y}_t = \bm{R}_{\rmnew} \bm{a}_{t,\rmnew}{\bm{a}_{t,*}}' {\bm{D}_*}' {\bm{E}_{\rmnew}} +  {\bm{E}_{\rmnew}}' \bm{D}_* \bm{a}_{t,*}{\bm{a}_{t,\rmnew}}' {\bm{R}_{\rmnew}}'$.
\begin{enumerate}
\item  The $\bm{Y}_t$'s are conditionally independent given $X_{k-1}$.

\item Using the fact that $\bm{a}_{t}$ are zero mean and $\cov(\bm{a}_t)$ is diagonal,  $\E\left(\frac{1}{\alpha}\sum_t \bm{Y}_t|X_{k-1}\right) = 0$ for all $X_{k-1} \in \Gamma_{k-1}$.

\item Using  the bound on $\zeta$ from Theorem \ref{thm1}, $\|\bm{Y}_t\| \le 2\sqrt{c r} \zeta_*^+ \gamma \gamma_{\rmnew}  \leq 2\sqrt{c r} \zeta_*^+ {\gamma}^2 \le  2$ holds w.p. one for all $X_{k-1} \in \Gamma_{k-1}$.
\end{enumerate}

Thus, under the same conditioning, $-2 \bm{I} \preceq \bm{Y}_t  \preceq 2 \bm{I}$ with w.p. one.

Thus, applying Corollary \ref{hoeffding_nonzero} with $\epsilon = \frac{c\zeta\lambda^-}{24}$, we get
\beq
\mathbb{P}\left(\lambda_{\min} \left(\frac{1}{\alpha} \sum_t \bm{Y}_t \right) \geq  \frac{-c\zeta\lambda^-}{24} \Big| X_{k-1} \right) \geq 1- c \exp \left( \frac{-\alpha c^2 \zeta^2(\lambda^-)^2} {8 \cdot 24^2 \cdot (4)^2}\right)  \ \text{for all $X_{k-1} \in  \Gamma_{k-1}$}
\label{lemma_add_A2}
\eeq

Combining (\ref{lemmabound_1}), (\ref{lemma_add_A1}) and (\ref{lemma_add_A2}) and using the union bound,
\[
\mathbb{P} \left(\lambda_{\min}(\bm{A}_k) \geq \lambda_{\rmnew}^-(1 - (\zeta_*^+)^2) - \frac{c\zeta\lambda^-}{12} \Big| X_{k-1} \in \Gamma_{k-1}\right) \geq 1-p_{a}(\alpha,\zeta)
\]
The first claim of the lemma follows by using $\lambda_{\rmnew}^- \ge \lambda^-$. 

Now consider $\bm{A}_{k,\perp} := \frac{1}{\alpha} \sum_t {\bm{E}_{\rmnew,\perp}}' \bm{\Phi}_{(0)} \bm{\ell}_t {\bm{\ell}_t}' \bm{\Phi}_{(0)} \bm{E}_{\rmnew,\perp}$. By their definitions, ${\bm{E}_{\rmnew,\perp}}' \bm{\Phi}_{(0)} \bm{\ell}_t = {\bm{E}_{\rmnew,\perp}}' \bm{D}_* \bm{a}_{t,*}$. Thus, $\bm{A}_{k,\perp} = \frac{1}{\alpha} \sum_t \bm{Z}_t$ with  $\bm{Z}_t={\bm{E}_{\rmnew,\perp}}' \bm{D}_* \bm{a}_{t,*} {\bm{a}_{t,*}}' {\bm{D}_*}' \bm{E}_{\rmnew,\perp}$ which is of size $(n-c)\times (n-c)$.
Using the same ideas as above we can show that $0 \preceq \bm{Z}_t \preceq r (\zeta_*^+)^2 {\gamma}^2 \bm{I} \preceq \zeta \bm{I}$ and $\E\left(\frac{1}{\alpha}\sum_t \bm{Z}_t|X_{k-1}\right) \preceq (\zeta_*^+)^2 \lambda^+ \bm{I}$.  Thus by Corollary \ref{hoeffding_nonzero} with $\epsilon = \frac{c\zeta\lambda^-}{24}$ 
the second claim follows.  Notice that $\frac{\lambda^+}{\lambda_{\rmnew}^-}\leq f$.

Now consider the $\bm{\mathcal{H}}_k$ term.  For ease of notation, define
\[
\Ltil_t = \bm{\Phi}_{(0)} \bm{\ell}_t.
\]
Using the expression for $\bm{\mathcal{H}_k}$ given in Definition \ref{defHk}, and noting that $\bm{EE}' + \bm{E}_{\perp}{\bm{E}_{\perp}}' = \I$ we get that
\[
\bm{\mathcal{H}}_k = \frac{1}{\alpha} \sum_t \Big( \bm{\Phi}_{(0)} \bm{e}_t {\bm{e}_t}' \bm{\Phi}_{(0)} - (\Ltil_t {\bm{e}_t}' \bm{\Phi}_{(0)} + \bm{\Phi}_{(0)} \bm{e}_t \Ltil_t{}')  + (\bm{F}_t + {\bm{F}_t}') \Big)
\]
where
\[
\bm{F}_t = \bm{E}_{\rmnew,\perp}{\bm{E}_{\rmnew,\perp}}' \Ltil_t \Ltil_t{}' \bm{E}_{\rmnew}{\bm{E}_{\rmnew}}'.
\]
Thus,
\begin{equation}\label{add_calH1}
\| \bm{\mathcal{H}}_k \|_2 \leq \bigg\| \frac{1}{\alpha} \sum_t \bm{e}_t {\bm{e}_t}' \bigg\|_2 + 2\bigg\| \frac{1}{\alpha} \sum_t  \Ltil_t {\bm{e}_t}'  \bigg\|_2 +  2\bigg\| \frac{1}{\alpha} \sum_t \bm{F}_t \bigg\|_2
\end{equation}

First note that by Lemma \ref{cslem} $\et$ satisfies \eqref{etdef0} if $X_{k-1}\in\Gamma_{k-1}$.

Next, we obtain high probability bounds on each of the three terms on the right hand side of (\ref{add_calH1}) using the Hoeffding corollaries.

Consider $\|\frac{1}{\alpha} \sum_t \bm{e}_t {\bm{e}_t}'\|_2$. Let $\bm{Z}_t = \bm{e}_t {\bm{e}_t}'$.

\begin{enumerate}
\item Conditioned on $X_{k-1}$, the various $\bm{Z}_t$'s in the summation are independent, for all $X_{k-1} \in \Gamma_{k-1}$.
\item Using Lemma \ref{cslem}
conditioned on $X_{k-1}\in\Gamma_{k-1}$, $\bm{0} \preceq \bm{Z}_t \preceq b_1 \bm{I}$ with probability one for all $X_{k-1} \in  \Gamma_{k-1}$. Here $b_1 :=\Big( \phi^+ (\zeta_*^+ \sqrt{r} \gamma + \zeta_{k-1}^+ \sqrt{c}\gamma_{\rmnew}) \Big)^2$.
\item  Using the expression for $\et$ in Lemma \ref{cslem} when $X_{k-1} \in \Gamma_{k-1}$,
\begin{align*}
 &  \frac{1}{\alpha}\sum_t \E \left[ \et{\et}' | X_{k - 1} \right] \\
&=   \frac{1}{\alpha}\sum_t \E \left[\Big( \bm{I}_{\mathcal{T}_t} [ ({\bm{\Phi}_{t})_{\mathcal{T}_t}}'(\bm{\Phi}_{t})_{\mathcal{T}_t}]^{-1}  {\bm{I}_{\mathcal{T}_t}}' \bm{\Phi}_{t} \bm{\ell}_t \Big)\Big(  \bm{I}_{\mathcal{T}_t} [ ({\bm{\Phi}_{t})_{\mathcal{T}_t}}'(\bm{\Phi}_{t})_{\mathcal{T}_t}]^{-1}  {\bm{I}_{\mathcal{T}_t}}' \bm{\Phi}_{t} \bm{\ell}_t \Big)' | X_{k - 1} \right]\\
&=    \frac{1}{\alpha}\sum_t \E\left[\Big(  \bm{I}_{\mathcal{T}_t} [ ({\bm{\Phi}_{t})_{\mathcal{T}_t}}'(\bm{\Phi}_{t})_{\mathcal{T}_t}]^{-1}  {\bm{I}_{\mathcal{T}_t}}' \bm{\Phi}_{t} \bm{\ell}_t {\lt}' \bm{\Phi}_{t}\bm{I}_{\mathcal{T}_t}[ ({\bm{\Phi}_{t})_{\mathcal{T}_t}}'(\bm{\Phi}_{t})_{\mathcal{T}_t}]^{-1}{\bm{I}_{\mathcal{T}_t}}' | X_{k - 1} \right] \\
&= \frac{1}{\alpha}\sum_t   \bm{I}_{\mathcal{T}_t} [ ({\bm{\Phi}_{t})_{\mathcal{T}_t}}'(\bm{\Phi}_{t})_{\mathcal{T}_t}]^{-1}  {\bm{I}_{\mathcal{T}_t}}' \bm{\Phi}_{t} (\bm{P}_{*}(\bm{\Lambda}_{t})_{*}{\bm{P}_{*}}' + \bm{P}_{\rmnew}(\bm{\Lambda}_{t})_{\new}{\bm{P}_{\rmnew}}' ) \bm{\Phi}_{t}\bm{I}_{\mathcal{T}_t}[ ({\bm{\Phi}_{t})_{\mathcal{T}_t}}'(\bm{\Phi}_{t})_{\mathcal{T}_t}]^{-1}{\bm{I}_{\mathcal{T}_t}}'\\
&= \frac{1}{\alpha}\sum_t   \bm{I}_{\mathcal{T}_t} [ ({\bm{\Phi}_{t})_{\mathcal{T}_t}}'(\bm{\Phi}_{t})_{\mathcal{T}_t}]^{-1}  {\bm{I}_{\mathcal{T}_t}}' \Big(\bm{D}_{*,k-1}(\bm{\Lambda}_{t})_{*}{\bm{D}_{*,k-1}}'+ \bm{D}_{\rmnew,k-1}(\bm{\Lambda}_{t})_{\new}{\bm{D}_{\rmnew,k-1}}' \Big) \bm{I}_{\mathcal{T}_t}[ ({\bm{\Phi}_{t})_{\mathcal{T}_t}}'(\bm{\Phi}_{t})_{\mathcal{T}_t}]^{-1}{\bm{I}_{\mathcal{T}_t}}'  \\
\end{align*}

When $k=1$ we can apply Lemma \ref{Dnew0_lem} to get that the $\|{\bm{D}_{\new}}'\I_{\T_t}\|_2 \leq \kappa_s^+$.  Then we apply Lemma \ref{blockdiag} with $\sigma^+ = (\phi^+)^2 \left( (\zeta_{*}^+)^2\lambda^+ +   (\kappa_{s}^+)^2 \lambda_{\rmnew}^+ \right)$.
This gives
\[
\bm{0} \preceq \E \left[  \sum_{t\in\mathcal{I}_{j,1}} \et{\et}' \Big| X_{0} \right] \preceq  8h^+(\phi^+)^2 \Big((\zeta_*^+)^2\lambda^+ + (\kappa_{s}^+)^2 \lambda_{\rmnew}^+ \Big)\I
\quad\text{ for all }X_{0} \in \Gamma_{0}.
\]
When $k\geq2$ we can apply Lemma \ref{blockdiag} with $\sigma^+ = (\phi^+)^2 \left( (\zeta_{*}^+)^2\lambda^+ +  (\zeta_{k-1}^+)^2 \lambda_{\rmnew}^+ \right)$ and the assumed bound $h^*(\alpha)\leq h^+\alpha$ to get that,
\[
\bm{0} \preceq \E \left[  \sum_{t\in\mathcal{I}_{j,k}} \et{\et}' \Big| X_{k - 1} \right] \preceq  8h^+(\phi^+)^2 \Big((\zeta_*^+)^2\lambda^+ + (\zeta_{k-1}^+)^2 \lambda_{\rmnew}^+ \Big)\I
\quad\text{ for all }X_{k - 1} \in \Gamma_{k-1}.
\]
\end{enumerate}

So define 
\[
b_{2,k} := \begin{cases} 
8h^+(\phi^+)^2 \Big((\zeta_*^+)^2\lambda^+ + (\kappa_{s}^+)^2\lambda_{\rmnew}^+\Big) & k=1 \\
8h^+(\phi^+)^2 \Big((\zeta_*^+)^2\lambda^+ + (\zeta_{k-1}^+)^2\lambda_{\rmnew}^+\Big) & k\geq2 .
\end{cases}
\]
Thus, applying Corollary \ref{hoeffding_nonzero} with $\epsilon = \frac{c\zeta\lambda^-}{24}$ and Fact \ref{event},
\beq
\mathbb{P} \left( \Big\|\frac{1}{\alpha} \sum_{t\in\mathcal{I}_{j,k}}  \bm{e}_t {\bm{e}_t}' \Big\|_2 \leq b_{2,k}  + \frac{c\zeta\lambda^-}{24} \Big| X_{k-1} \in \Gamma_{k-1} \right) \geq 1- n \exp\left(\frac{-\alpha c^2 \zeta^2 (\lambda^-)^2}{ 8 \cdot 24^2 {b_1}^2}\right).
\label{add_etet}
\eeq

Consider $\big\| \frac{1}{\alpha} \sum_t  \Ltil_t {\bm{e}_t}'  \big\|_2$. Let $\bm{Z}_t: = \Ltil_t {\bm{e}_t}' $.
\begin{enumerate}
\item Conditioned on $X_{k-1}$, the various $\bm{Z}_t$'s used in the summation are mutually independent, for all $X_{k-1} \in \Gamma_{k-1}$.

\item For  $X_{k-1} \in \Gamma_{k-1}$, $\|\bm{Z}_t\|_2 = \|\Ltil_t {\bm{e}_t}' \|_2 \leq \Big(\zeta_*^+\sqrt{r}\gamma + \sqrt{c}\gamma_{\rmnew} \Big) \Big( \phi^+ (\zeta_*^+ \sqrt{r} \gamma + \zeta_{k-1}^+ \sqrt{c}\gamma_{\rmnew}) \Big) := b_3$
holds with probability one.

\item Using the same bounds, when $X_{0} \in \Gamma_{0}$, (notice that this is the $k=1$ case)
\begin{align*}
&\bigg\| \E \bigg[   \frac{1}{\alpha}\sum_t \Ltil_t {\bm{e}_t}'  \ \big| \ X_{0} \bigg] \bigg\|_2  \\
=& \bigg\|  \E\bigg[ \frac{1}{\alpha}\sum_t\bm{\Phi}_{(0)}(\bm{P}_{*}\bm{a}_{t,*} + \bm{P}_{\rmnew} \bm{a}_{t,\rmnew}) (\bm{P}_{*} \bm{a}_{t,*} + \bm{P}_{\rmnew} \bm{a}_{t,\rmnew})' {\bm{\Phi}_{t}}'\bm{I}_{\mathcal{T}_t}[{(\bm{\Phi}_{t})_{\mathcal{T}_t}}'(\bm{\Phi}_{t})_{\mathcal{T}_t}]^{-1} {\bm{I}_{\mathcal{T}_t}}'\bm{\Phi}_{(0)}   \ \big| \ X_{0} \bigg]\bigg\|_2 \\
=& \bigg\| \E \bigg[   \frac{1}{\alpha} \sum_t (\bm{D}_* \bm{a}_{t,*} + \bm{D}_{\new} \bm{a}_{t,\new})(\bm{D}_{*} \bm{a}_{t,*} + \bm{D}_{\new} \bm{a}_{t,\rmnew})' \bm{I}_{\mathcal{T}_t}[{(\bm{\Phi}_{t})_{\mathcal{T}_t}}'(\bm{\Phi}_{t})_{\mathcal{T}_t}]^{-1} {\bm{I}_{\mathcal{T}_t}}'\bm{\Phi}_{(0)}  \ \big| \ X_{0} \bigg] \bigg\|_2 \\
=& \bigg\| \frac{1}{\alpha}\sum_t \left[ \Big(   \bm{D}_*  (\bm{\Lambda}_t)_*  {\bm{D}_{*}}'  + \bm{D}_{\rmnew} (\bm{\Lambda}_{t})_{\new}  {\bm{D}_{\new}}' \Big) \bm{I}_{\mathcal{T}_t}[{(\bm{\Phi}_{t})_{\mathcal{T}_t}}'(\bm{\Phi}_{t})_{\mathcal{T}_t}]^{-1} {\bm{I}_{\mathcal{T}_t}}' \right] \bigg\|_2\\
\leq & (\zeta_{*}^+)^2\lambda^+ + \kappa_{s}^+\lambda_{\new}^+.
\end{align*}
And when $k\geq2$,
\begin{align*}
& \bigg\| \frac{1}{\alpha}\sum_t \left[ \Big(   \bm{D}_*  (\bm{\Lambda}_t)_*  {\bm{D}_{*,k-1}}'  + \bm{D}_{\rmnew} (\bm{\Lambda}_{t})_{\new}  {\bm{D}_{\rmnew,k-1}}' \Big) \bm{I}_{\mathcal{T}_t}[{(\bm{\Phi}_{t})_{\mathcal{T}_t}}'(\bm{\Phi}_{t})_{\mathcal{T}_t}]^{-1} {\bm{I}_{\mathcal{T}_t}}' \right] \bigg\|_2\\
\leq& \sqrt{\lambda_{\max}\left(\frac{1}{\alpha}\sum_t \Big(   \bm{D}_*  (\bm{\Lambda}_t)_*  {\bm{D}_{*,k-1}}'  + \bm{D}_{\rmnew} (\bm{\Lambda}_{t})_{\new}  {\bm{D}_{\rmnew,k-1}}' \Big)\Big(   \bm{D}_*  (\bm{\Lambda}_t)_*  {\bm{D}_{*,k-1}}'  + \bm{D}_{\rmnew} (\bm{\Lambda}_{t})_{\new}  {\bm{D}_{\rmnew,k-1}}' \Big)' \right)} \\
& \sqrt{\lambda_{\max}\left( \frac{1}{\alpha}\sum_t \Big(\bm{I}_{\mathcal{T}_t}[{(\bm{\Phi}_{t})_{\mathcal{T}_t}}'(\bm{\Phi}_{t})_{\mathcal{T}_t}]^{-1} {\bm{I}_{\mathcal{T}_t}}'\Big)\Big(\bm{I}_{\mathcal{T}_t}[{(\bm{\Phi}_{t})_{\mathcal{T}_t}}'(\bm{\Phi}_{t})_{\mathcal{T}_t}]^{-1} {\bm{I}_{\mathcal{T}_t}}'\Big)' \right)  } \\
\leq& \left( (\zeta_*^+)^2 \lambda^+ + \zeta_{k-1}^+\lambda_{\rmnew}^+  \right)\left(\sqrt{8h^+}\phi^+\right).
\end{align*}
The first inequality is Cauchy-Schwarz for a sum of matrices.  This can be found as Lemma \ref{CSmat} in the appendix.
The last line uses Lemma $\ref{blockdiag}$ with $\sigma^+ = (\phi^+)^2$.
So define
\[
b_{4,k} := 
\begin{cases} 
(\zeta_{*}^+)^2\lambda^+ + \kappa_{s}^+\lambda_{\new}^+ & k=1 \\
\left( (\zeta_*^+)^2 \lambda^+ + \zeta_{k-1}^+\lambda_{\rmnew}^+  \right)\left(\sqrt{8h^+}\phi^+\right) & k\geq2.
\end{cases}
\]

By Corollary \ref{hoeffding_rec} with $\epsilon = \frac{c\zeta\lambda^-}{24}$ and Fact \ref{event},
\begin{equation*}
\mathbb{P} \left(  \Big\| \frac{1}{\alpha} \sum_{t\in\Ijk}  \Ltil_t {\bm{e}_t}'  \Big\|_2 \leq b_{4,k} + \frac{c\zeta\lambda^-}{24} \Bigg| X_{k-1} \in \Gamma_{k-1} \right) \geq 1 - n \exp\left(  \frac{-\alpha c^2\zeta^2(\lambda^-)^2}{8 \cdot 24^2 {b_3}^2} \right).
\end{equation*}
and so
\begin{equation}\label{Ltet}
\mathbb{P} \left(  2\Big\| \frac{1}{\alpha} \sum_t  \Ltil_t {\bm{e}_t}' \Big\|_2 \leq 2b_{4,k} + \frac{c\zeta\lambda^-}{12} \Bigg| X_{k-1} \in \Gamma_{k-1} \right) \geq 1 - n \exp\left(  \frac{-\alpha c^2\zeta^2(\lambda^-)^2}{8 \cdot 24^2 {b_3}^2} \right).
\end{equation}

\end{enumerate}

Finally, consider  $\big\| \frac{1}{\alpha} \sum_t \bm{F}_t \big\|_2$.
\begin{enumerate}
\item Conditioned on $X_{k-1}$, the $\bm{F}_t$'s  are mutually independent, for all $X_{k-1} \in  \Gamma_{k-1}$.
\item For $X_{k-1} \in  \Gamma_{k-1}$, 
\begin{align*}
\|\bm{F}_t\|_2 &= \| \bm{E}_{\rmnew,\perp}{\bm{E}_{\rmnew,\perp}}' \Ltil_t \Ltil_t{}' \bm{E}_{\rmnew}{\bm{E}_{\rmnew}}' \|_2 \\
&= \| \bm{E}_{\rmnew,\perp}{\bm{E}_{\rmnew,\perp}}'\bm{D}_*\bm{a}_t {\bm{a}_t}' \Pj \bm{\Phi}_{(0)}\bm{E}_{\rmnew}{\bm{E}_{\rmnew}}' \|_2 \\
&\leq \| \bm{D}_* \bm{a}_t{\bm{a}_t}'  \|_2 \\
&\leq \zeta_*^+ (\sqrt{r}\gamma  )^2 := b_5
\end{align*}
 holds with probability 1.
Here and below, notice that $\bm{E}_{\rmnew,\perp}$ nullifies $\bm{D}_{\rmnew}$.

\item For $X \in \Gamma_{k-1}$,
\begin{align*}
\bigg\| \E \Big[ \frac{1}{\alpha} \sum_t \bm{F}_t  \ \big| \ X_{k-1} \Big] \bigg\|_2 &  = \bigg\| \E \Big[ \frac{1}{\alpha} \sum_t \bm{E}_{\rmnew,\perp} {\bm{E}_{\rmnew,\perp}}' \Ltil_t \Ltil_t{}' \bm{E}_{\rmnew}{\bm{E}_{\rmnew}}' \ \big| \ X_{k-1} \Big] \bigg\|_2 \\
&\leq \bigg\| \E \Big[ \frac{1}{\alpha} \sum_t \bm{E}_{\rmnew,\perp}{\bm{E}_{\rmnew,\perp}}'\Ltil_t \Ltil_t{}' \ \big| \ X_{k-1} \Big] \bigg\|_2  \\
&= \bigg\|\frac{1}{\alpha}\sum_t  \bm{E}_{\rmnew,\perp} {\bm{E}_{\rmnew,\perp}}' (\bm{D}_*\bm{\Lambda}_t {\bm{D}_*}' + \bm{D}_{\rmnew}\bm{\Lambda}_{(j),\rmnew}\bm{D}_{\rmnew}') \bigg\|_2 \\
&= \bigg\|\frac{1}{\alpha}\sum_t  \bm{E}_{\rmnew,\perp} {\bm{E}_{\rmnew,\perp}}' (\bm{D}_*\bm{\Lambda}_t {\bm{D}_*}' ) \bigg\|_2 \\
&\leq \bigg\|\frac{1}{\alpha}\sum_t  (\bm{D}_*\bm{\Lambda}_t {\bm{D}_*}' ) \bigg\|_2 \\
&\leq (\zeta_*^+)^2\lambda^+  := b_6
\end{align*}

\end{enumerate}

Applying Corollary \ref{hoeffding_rec} with $\epsilon = \frac{c\zeta\lambda^-}{24}$ and Fact \ref{event},
\begin{equation}\label{Ft}
\mathbb{P} \left( 2 \Big\|  \frac{1}{\alpha}\sum_t \bm{F}_t  \Big\|_2 \leq 2b_6 + \frac{c\zeta\lambda^-}{12} \Bigg| X_{k-1} \in \Gamma_{k-1} \right) \geq 1 - n \exp\left(  \frac{-\alpha c^2\zeta^2(\lambda^-)^2}{8 \cdot 24^2 {b_5}^2} \right)
\end{equation}


Using (\ref{add_calH1}), (\ref{add_etet}), (\ref{Ltet}) and (\ref{Ft}) and the union bound,  for any $X_{k-1} \in  \Gamma_{k-1}$,
\begin{align}
&\mathbb{P} \left(\|\bm{\mathcal{H}}_k\|_2 \leq b + \frac{5c\zeta\lambda^-}{24} \Big|X_{k-1}\right)\geq \nn \\
&  1-n \exp\left(\frac{-\alpha c^2 \zeta^2 (\lambda^-)^2}{8 \cdot 24^2 {b_1}^2}\right)- n \exp\left(\frac{-\alpha c^2\zeta^2 (\lambda^-)^2}{8\cdot 24^2 \cdot  {b_3}^2}\right)- n \exp\left(\frac{-\alpha c^2 \zeta^2 (\lambda^-)^2 }{8 \cdot 24^2 {b_5}^2}\right)
\label{normHk}
\end{align}
where $b_k := b_{2,k} +2b_{4,k}+ 2b_6$.


Applying the bounds assumed in Theorem \ref{thm1} and $\zeta_{k-1}^+ \leq 1$, we get that the probability is at least $1 - p_c(\alpha,\zeta)$.

\end{proof}

\section{Proof of Lemma \ref{probsupch}}\label{probpf}

\begin{proof}
To show that the $\mathcal{T}_t$ satisfy Example \ref{sby2}, we will show:
\begin{enumerate}
\item in an interval $\mathcal{J}_{u}$, the support changes fewer than $\frac{n}{2s}$ times;
\item the support changes at least once every $\frac{\alpha}{200}$ instants with probability greater than $1 - \frac{n^{-10}}{2}$;
\item when it changes, the support moves by at least $\frac{s}{2}$ and not more than $2s$ indices with probability greater than $1 - \frac{n^{-10}}{2}$.
\end{enumerate}

To see 1), observe that the object moves at most $\alpha$ times during an interval $\mathcal{J}_{u}$.  Then the assumed bound $s \leq \frac{n}{2\alpha}$ implies that the the support changes fewer than $\frac{n}{2s}$ times during an interval $\mathcal{J}_{u}$.  So 1) occurs with probability 1.

For 2) we have by the bound in \cite{longest_run}
\begin{align*}
&\mathbb{P}\left(\text{The object moves at least once every } \frac{\alpha}{200} \text{ instants in the interval } \mathcal{J}_{u}\right)\\
& =  \mathbb{P}\left( \text{The bit sequence } \theta_{(u-1)\alpha} \dots \theta_{u\alpha-1} \text{ does not contain a sequence of } \frac{\alpha}{200} \text{ consecutive zeros} \right) \\
& \geq  \left( 1 - (1-q)^{\frac{\alpha}{200}} \right)^{\alpha - \frac{\alpha}{200} + 1}  \geq \left(1 - (1-q)^{\frac{\alpha}{200}} \right)^{\alpha}  .
\end{align*}
We need the object to move at least once every $\frac{\alpha}{200}$ time instants  in {\em every} interval $\mathcal{J}_{u}$.  We have
\begin{align*}
\mathbb{P}\left(\text{The object moves at least once every } \frac{\alpha}{200} \text{ instants in {\em every} interval } \mathcal{J}_{u}\right) &\geq  \left(1 - (1-q)^{\frac{\alpha}{200}} \right)^{\lceil \frac{t_{\max}}{\alpha} \rceil \alpha} \\
&\geq \left(1 - (1-q)^{\frac{\alpha}{200}} \right)^{( \frac{t_{\max}}{\alpha} + 1) \alpha} \\
&\geq 1 - (t_{\max} + \alpha) (1 - q)^{\frac{\alpha}{200}}.
\end{align*}
This probability will be greater than $1 - \frac{n^{-10}}{2}$ if
\[
q \geq 1 - \left(  \frac{n^{-10}}{2(t_{\max}+\alpha)} \right)^{\frac{200}{\alpha}} 
\]

%

To prove 3), consider the probability of having motion of at least $\frac{s}{2}$ indices whenever the object moves. This will happen if $\nu_t \geq -0.1s$ for $t=1,\dots,t_{\max}$. Also, if $\nu_t \leq .1s$, then the object will move by fewer than $.7s\leq 2s$ indices.  Using a standard Gaussian tail bound,
\begin{align*}
\mathbb{P}(|\nu_t| \leq 0.1s)^{t_{\max}}
&\geq \left( 1 - \frac{2\exp\left(\frac{-(0.1s)^2}{2\sigma^2}\right)}{\frac{0.1s}{\sigma}\sqrt{2\pi}} \right)^{t_{\max}} \\
&= \left( 1 - \frac{20\sigma\exp\left(\frac{-(0.1s)^2}{2\sigma^2}\right)}{s \sqrt{2\pi}} \right)^{t_{\max}} \\ 
&= \left( 1 - \frac{20\sqrt{\frac{\rho}{\log(n)}}\exp\left(\frac{-0.01\log(n)}{2\rho }\right)}{ \sqrt{2\pi}} \right)^{t_{\max}} \\
&\geq 1 - t_{\max}\frac{20}{\sqrt{2\pi}}\sqrt{\frac{\rho}{\log(n)}} \exp\left(\frac{-.01}{2\rho}\log(n)\right) \\
&=  1 - t_{\max}\frac{20\sqrt{\rho}}{\sqrt{2\pi}} \frac{n^{\frac{-0.005}{\rho}}}{\sqrt{\log(n)}} \\
&\geq 1 - \frac{20\sqrt{\rho}}{\sqrt{2\pi}} \frac{n^{10 - \frac{0.005}{\rho}}}{\sqrt{\log(n)}} \\
&\geq 1 - \frac{n^{-10}}{2}.
\end{align*}
The last line uses the bound $\rho \leq \frac{1}{4000}$, and for simplicity we assume $n\geq 3$, so that $\log(n)\geq 1$.

Finally, by the union bound
\[
\mathbb{P}\Big( 2) \text{ and } 3) \text{ hold} \Big) =1 - \Pr \big( 2) \text{ or } 3) \text{ does not hold} \big) \geq 1 - 2\frac{n^{-10}}{2} = 1 - n^{-10}
\]

The only modification to the proof of Theorem \ref{thm1} is in the definition of the random variable $\bm{X}_{k-1}$.  See Definition \ref{Gamma_k}.

\end{proof}

\section{Simulations}\label{simulation}

\begin{figure}[ht]
	\centering
	\includegraphics{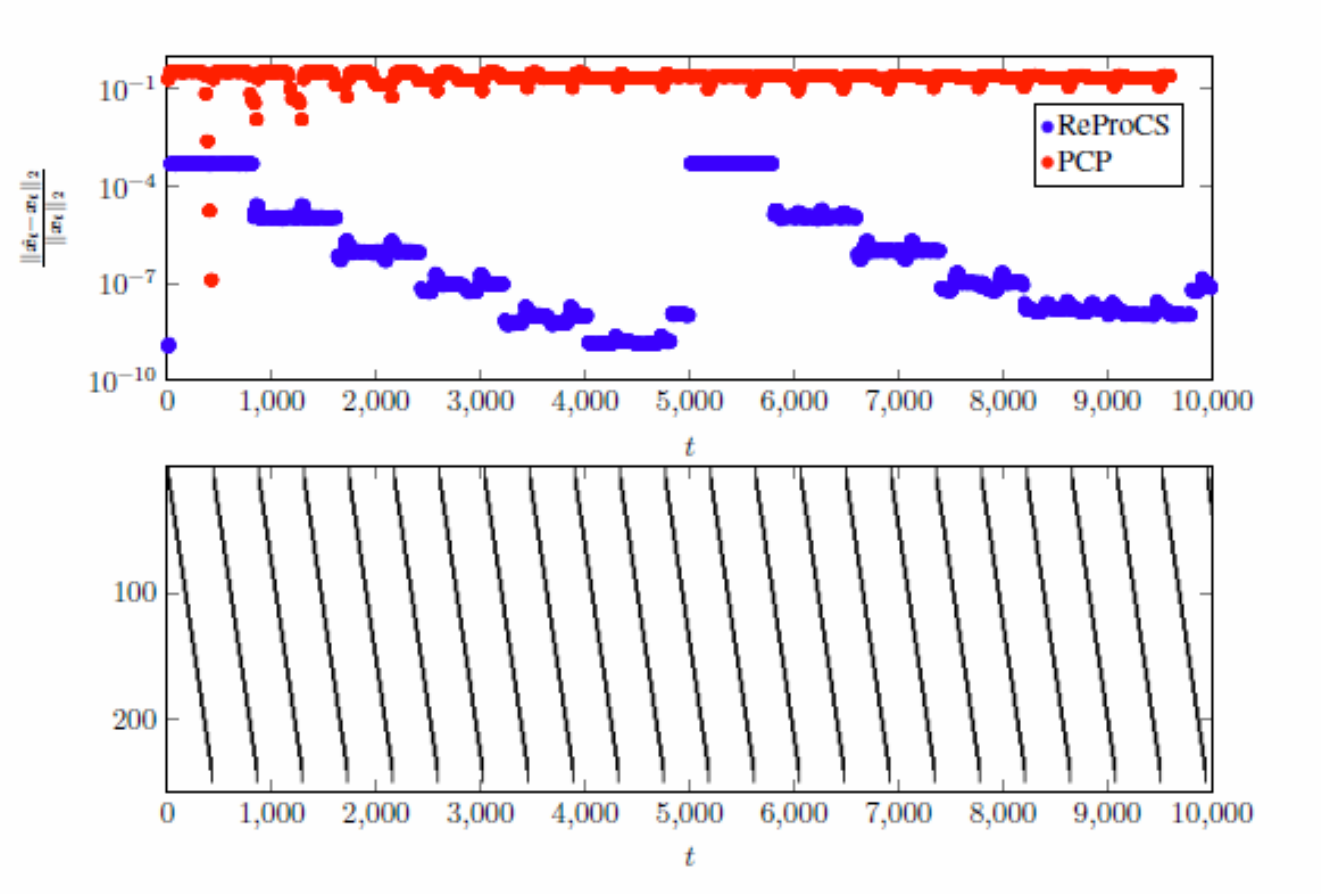}
	\caption{A comparison of ReProCS and PCP for {\bf correlated} supports of $\xt$ shown with the support pattern of $\bm{X}$.
  Results are averaged over 100 Monte Carlo trials.  The $y$ axis scale is logarithmic.
\label{fig:cor}}
\end{figure}

In this section we provide some simulations that demonstrate the result we have proven above.

The data for Figure \ref{fig:cor} was generated as follows.
We chose $n = 256$ and $t_{\max} = 10000$.  The size of the support of $\xt$ was 20.
Each non-zero entry of $\xt$ was drawn uniformly at random between 2 and 6 independent of other entries and other times $t$.  In Figure \ref{fig:cor} the support of $\xt$ changes as assumed in Theorem \ref{thm1}.
For the ReProCS algorithm, we choose $K = 6$ and $\alpha = 800$.
So the support of $\xt$ changes by $\frac{s}{2} = 10$ indices every $\left\lfloor \frac{\alpha}{44} \right\rfloor = 18$ time instants.
When the support of $\xt$ reaches the bottom of the vector, it starts over again at the top.
This pattern can be seen in the bottom half of the figure which shows the sparsity pattern of the matrix $\bm{X} = [\bm{x}_1, \dots, \bm{x}_{t_{\max}}]$.

To form the low dimensional vectors $\lt$,  we started with an $n \times r$ matrix of i.i.d. Gaussian entries and orthonormalized the columns using Gram-Schmidt.
The first $r_0 = 10$ columns of this matrix formed $\bm{P}_{(0)}$, the next 2 columns formed $\bm{P}_{(1),\rmnew}$, and the last 2 columns formed $\bm{P}_{(2),\rmnew}$
We show two subspace changes which occur at $t_1 = 25$ and $t_2 = 5001$.
The entries of  $\bm{a}_{t,*}$ were drawn uniformly at random between -5 and 5,
and the entries of $\bm{a}_{t,\rmnew}$ were drawn uniformly at random between -.04 and .04 for the first $K\alpha$ frames after being introduced,
and between -5 and 5 afterwards.
Entries of $\bm{a}_t$ were independent of each other and of the other $\bm{a}_t$'s.

For this simulated data we compare the performance of ReProCS and PCP.
For the initial subspace estimate $\hat{\bm{P}}_{(0)}$, we used $\bm{P}_{(0)}$ plus some small Gaussian noise and then obtained orthonormal columns.
We set $\alpha = 800$ and $K = 6$, so  the  new directions' variance increases after $t_1 + K\alpha = 4825$ and $t_2 + K\alpha = 9801$.
At this point the error seen by the sparse recovery step $\bm{\Phi}_t\lt$ increases,
so the error in estimating $\xt$ also increases.
For the PCP algorithm,  we perform the optimization every $\alpha$ time instants using all of the data up to that point.  So the first time PCP is performed on $[ \bm{m}_1 , \dots , \bm{m}_{\alpha}]$ and the second time it is performed on $[ \bm{m}_1 , \dots , \bm{m}_{2\alpha}]$ and so on.

Figure \ref{fig:cor} illustrates the result we have proven.  That is ReProCS takes advantage of the initial subspace estimate and slow subspace change (including the bound on $\gamma_{\rmnew}$) to handle the case when the supports of $\xt$ are correlated in time.
Notice how the ReProCS error increases after a subspace change, but decays exponentially with each projection PCA step.
For this data, the PCP program fails to give a meaningful estimate for all but a few times.
Compare this to Figure \ref{fig:rand} where the only change in the data is that the support of $\bm{X}$ is chosen uniformly at random from all sets of size $\frac{s t_{\max}}{n}$ (as assumed in \cite{rpca}).  Thus the total sparsity of the matrix $\bm{X}$ is the same for both figures.  In Figure \ref{fig:rand}, ReProCS performs almost the same, while PCP does substantially better than in the case of correlated supports.

\begin{figure}[ht]
	\centering
	\includegraphics{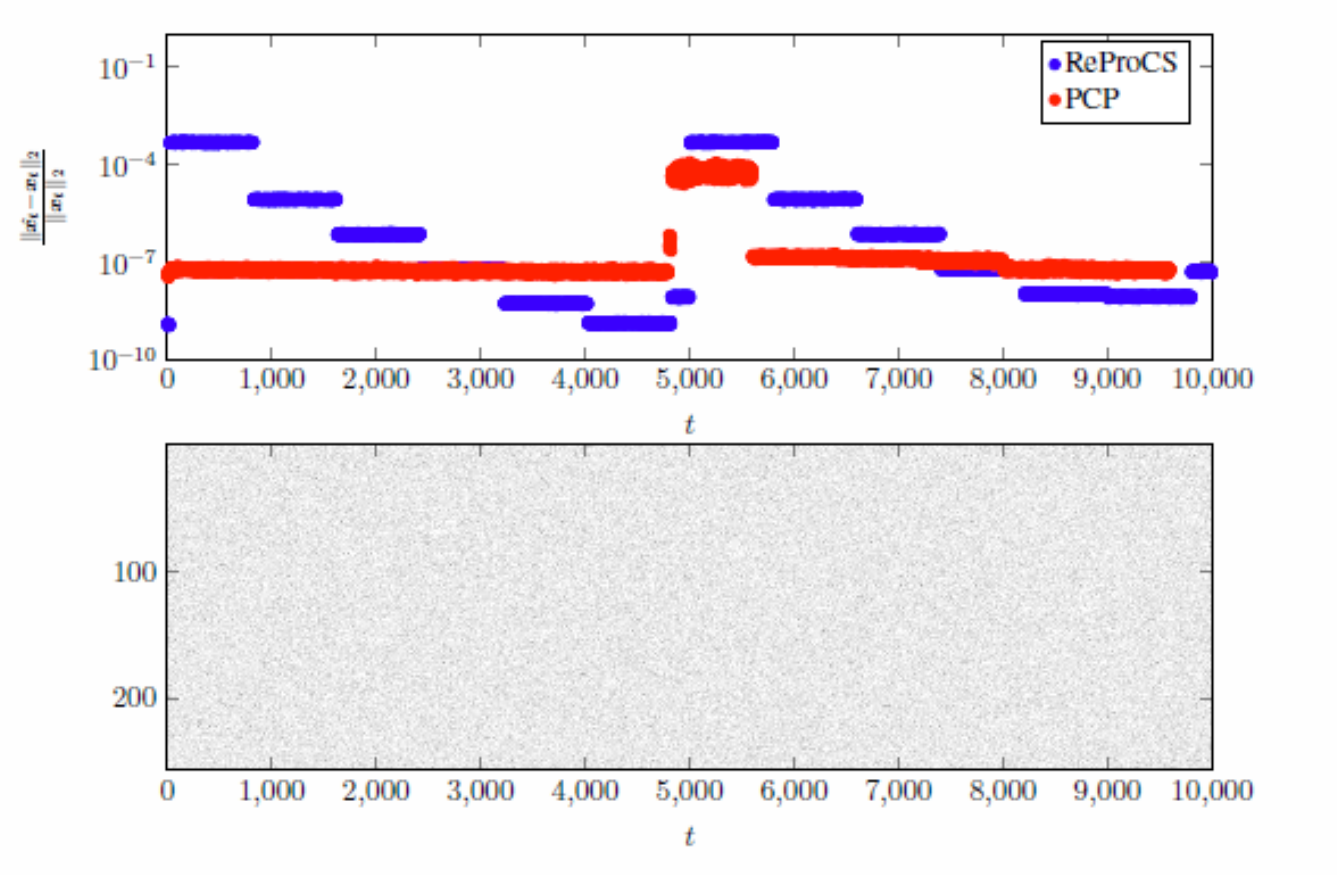}
	\caption{A comparison of ReProCS and PCP for {\bf uniformly random} support of $\bm{X}$ shown with the support pattern of $\bm{X}$.
  Results are averaged over 100 Monte Carlo trials.  The $y$ axis scale is logarithmic.
\label{fig:rand}}
\end{figure}

\section{Conclusions and Future Work}\label{conclusions}

In this paper we showed that with an accurate estimate of the initial subspace, denseness and slow subspace change assumptions on the vectors $\lt$, sparsity and support change assumptions on the vectors $\xt$, and other minor assumptions, with knowledge of model parameters $t_j$ and $c_{j,\rmnew}$, the online algorithm ReProCS will, with high probability, recover from $\mt = \xt + \lt$ the vectors $\xt$ and $\lt$ within a fixed, small tolerance.  Furthermore, the support of $\xt$ will be recovered exactly, and the subspace where the $\lt$'s lie will be accurately estimated within a finite delay of the the subspace change.  In future work we would like to extend our subspace change model to allow for the deletion of directions and relax the independence assumption on the $\lt$'s.  We also hope to explore the case where the support of $\xt$ is not exactly recovered and the case where the measurements $\mt$ also contain noise.

\appendices
\renewcommand{\thetheorem}{\thesection.\arabic{theorem}}


\section{Proof of Theorem \ref{thm1}}\label{thmpf}

\begin{definition}
Define the set $\check{\Gamma}_{j,k}$ as follows:
\begin{align*}
\check{\Gamma}_{j,k} &:= \{ X_{j,k} : \zeta_{j,k} \leq \zeta_{k}^+ \text{ and } \hat{\mathcal{T}}_t = \mathcal{T}_t  \text{ for all } t \in \mathcal{I}_{j,k} \} \\
\check{\Gamma}_{j,K+1} &:= \{ X_{j+1,0} :  \hat{\mathcal{T}}_t = \mathcal{T}_t  \text{ for all } t \in \mathcal{I}_{j,K+1} \}
\end{align*}

\end{definition}

\begin{definition}\label{Gamma_def}
Recursively define the sets $\Gamma_{j,k}$ as follows:
\begin{align*}
\Gamma_{1,0} &:= \{ X_{1,0}: \zeta_{1,*} \leq r_0 \zeta \ \text{and} \ \hat{\mathcal{T}}_t = \mathcal{T}_t \ \text{for all} \ t\in [t_{\mathrm{train}}+1: t_1 -1]\} \\
\Gamma_{j,0} &:=  \{X_{j,0}: \zeta_{j',*} \le \zeta_{j',*}^+ \ \text{for all} \ j' = 1, 2, \dots, j \ \text{and} \ \hat{\mathcal{T}}_t = \mathcal{T}_t \  \text{for all} \ t \le t_{j-1} \} \\
\Gamma_{j,k} &:=\Gamma_{j,k-1} \cap \check{\Gamma}_{j,k} \ k=1,2,\dots K+1
\end{align*}

\end{definition}


\begin{proof}[\textbf{Proof of Theorem \ref{thm1}}]

The theorem is a direct consequence of Lemmas \ref{expzeta}, \ref{cslem}, \ref{mainlem}, and Fact \ref{zetafact}.  
Observe that $\mathbb{P}(\Gamma_{j+1,0}^e|\Gamma_{j,0}^e) \ge \mathbb{P}(\check\Gamma_{j,1}^e, \dots \check\Gamma_{j,K+1}^e | \Gamma_{j,0}^e) = \prod_{k=1}^{K+1} \mathbb{P}(\check\Gamma_{j,k}^e|\Gamma_{j,k-1}^e)$. Also, since $\Gamma_{j+1,0} \subseteq \Gamma_{j,0}$, using the chain rule 
$\mathbb{P}(\Gamma_{J+1,0}^e | \Gamma_{1,0}^e) = \prod_{j=1}^{J}  \mathbb{P}(\Gamma_{j+1,0}^e | \Gamma_{j,0}^e)$.
Thus,
\[
\mathbb{P}(\Gamma_{J+1,0}^e | \Gamma_{1,0}^e) \ge \prod_{j=1}^{J}  \prod_{k=1}^{K+1} \mathbb{P}(\check\Gamma_{j,k}^e|\Gamma_{j,k-1}^e)
\]
Using Lemma \ref{mainlem}, we get
\[
\mathbb{P}(\Gamma_{J+1,0}^e| \Gamma_{1,0}) \geq {p}(\alpha,\zeta)^{KJ}.
\]
Also, $\mathbb{P}(\Gamma_{1,0}^e)=1$. This follows by the assumption on $\hat{\bm{P}}_0$ and Lemma \ref{cslem}. Thus, $\mathbb{P}(\Gamma_{J+1,0}^e) \geq {p}(\alpha,\zeta)^{KJ}$.

Using the lower bound on $\alpha$ assumed in the theorem, we get that
\[
\mathbb{P}(\Gamma_{J+1,0}^e) \geq {p}(\alpha,\zeta)^{KJ} \geq 1- n^{-10}.
\]

The event $\Gamma_{J+1,0}^e$ implies that $\hat{\mathcal{T}}_t = \mathcal{T}_t$ and $\bm{e}_t$ satisfies (\ref{etdef0}) for all $t < t_{\max}$. The definition of $\Gamma_{J+1,0}^e$  also implies that all the bounds on the subspace error hold. Using these, $\|\bm{a}_{t,\rmnew}\|_2 \le \sqrt{c} \gamma_{\rmnew}$ for $t \in [t_j, t_j + K\alpha]$, and $\|\bm{a}_t\|_2 \le \sqrt{r} \gamma$ for all $t$, $\Gamma_{J+1,0}^e$ implies that the bound on $\|\bm{e}_t\|_2$ holds.

Thus, all conclusions of the the result hold with probability at least $1- n^{-10}$.
\end{proof}

\section{Proof of Lemma \ref{expzeta} } \label{pfoflem1}

\begin{proof}
Recall that $\kappa_s^+ := 0.0215$, $\phi^+ := 1.2$, and $h^+ := \frac{1}{200}$, so we can make this substitution directly.
Notice that $\zeta_{j,k}^+$ is an increasing function of $\zeta_{j,*}^+$, and $ \zeta, c, f, g$.
Therefore we can use upper bounds on each of these quantities to get an upper bound on $\zeta_{j,k}^+$.
 From the bounds assumed on $\zeta$ in Theorem \ref{thm1} and $\zeta_{j,*}^+ := (r_0 + (j-1)c)\zeta$ we get
\begin{itemize}
\item $\zeta_{j,*}^+ \leq 10^{-4} $ and $c\zeta \leq 10^{-4}$
\item $\zeta_{j,*}^+ f \leq 1.5 \times 10^{-4}$
\item $\ds \frac{\zeta_{j,*}^+}{c\zeta} = \frac{(r_0+(j-1)c)\zeta}{c\zeta} \leq \frac{r_0 +(J-1)c}{c} = \frac{r}{c}\leq r$ (Without loss of generality we can assume that $c \geq 1$ because if $c=0$ then there is no subspace estimation problem to be solved. $c=0$ is the trivial case where all conclusions of Theorem \ref{thm1} will hold just using Lemma \ref{cslem}.)
\item $\zeta_{j,*}^+ f  \leq \zeta_{j,*}^+ f r \leq r^2 f \zeta \leq 1.5 \times 10^{-4}$
\end{itemize}
First we prove by induction that $\zeta_{j,k}^+ \leq \zeta_{j,k-1}^+  \leq 1$ for all $k\geq 1$.  Notice that $\zeta_{j,0}^+ =1$ by definition.
\begin{itemize}
\item Base cases: ($k=1$)  Using the above bounds we get that $\zeta_{j,1}^+ < 0.1 < 1 = \zeta_{j,0}^+$.
This also proves the first claim of the lemma for the $k=1$ case.

($k=2$) Using the above bounds and the bound on $\zeta_{j,1}^+$, we get that $\zeta_{j,2}^{+} \leq 0.06$.

\item For the induction step, assume that $\zeta_{j,k-1}^+ \leq \zeta_{j,k-2}^+$.  Then because $\zeta_{j,k}^+$ is increasing in $\zeta_{j,k-1}^+$ we get that $\zeta_{j,k}^+ \leq f_{\mathrm{inc}}(\zeta_{j,k-1}^+) \leq f_{\mathrm{inc}}(\zeta_{j,k-2}^+) = \zeta_{j,k-1}^+$.  Where $f_{\mathrm{inc}}$ represents the increasing function.
\end{itemize}
To prove the lemma we apply the above bounds including $\zeta_{k-1}^+ \leq 0.1$ and get that
\begin{align*}
\zeta_{j,k}^+  \leq \zeta_{k-1}^{+}(0.72)  + c\zeta(0.23) &= \zeta_0^{+} (0.72)^{k} + \sum_{i = 0}^{k-1}(0.72)^i(0.23) c\zeta \\
&\leq \zeta_0^{+} (0.72)^{k} + \sum_{i = 0}^{\infty}(0.72)^i(0.23) c\zeta \\
&\leq 0.72^k + 0.83 c\zeta
\end{align*}

\end{proof}

\section{Proof of Lemma \ref{cslem}}\label{pfcslem}

The proof of Lemma \ref{cslem} uses the subspace error bounds $\zeta_{*} \leq \zeta_{*}^{+}$ and $\zeta_{k-1} \leq \zeta_{k-1}^{+}$ that hold when $X_{j,k-1}\in\Gamma_{j,k-1}$ to obtain bounds on the restricted isometry constant of the sparse recovery matrix $\Phi_t$, and the sparse recovery error $\| \bm{b}_t \|_2$.  Applying the theorem in \cite{candes_rip} and the assumed bounds on $\zeta$ and $\gamma$, the result follows.

\begin{definition}
Define $\kappa_{s,*} := \kappa_s(\bm{P}_{(J)})$ and $\kappa_{s,\rmnew} := \max_{j} \kappa_s(\bm{P}_{(j),\rmnew})$.
\end{definition}
Recall that Theorem \ref{thm1} assumes that $\kappa_{2s,*} \leq 0.3$ and  $\kappa_{2s,\rmnew} \leq 0.02$.

\begin{lem} \cite[Lemma 2.10]{ReProCS_IT}\label{lemma0}\label{hatswitch}
Suppose that $\bm{P}$, $\Phat$ and $\bm{Q}$ are three basis matrices. Also, $\bm{P}$ and $\Phat$ are of the same size, $\bm{Q}' \bm{P} = 0$ and $\|(\I-\Phat \Phat{}' ) \bm{P} \|_2 = \zeta_*$. Then,
\begin{enumerate}
  \item $\|(I-\Phat\Phat{}')\bm{P}\bm{P}'\|_2 =\|( \I - \bm{P}\bm{P}' ) \Phat \Phat{}'\|_2 =  \|( \I - \bm{P} \bm{P}' ) \Phat \|_2 = \| ( \I - \Phat \Phat{}' ) \bm{P}\|_2 =  \zeta_*$
  \item $\|\bm{P} \bm{P}' - \Phat \Phat{}'\|_2 \leq 2 \|(\I-\Phat \Phat{}')\bm{P}\|_2 = 2 \zeta_*$
  \item $\|\Phat{}' \bm{Q}\|_2 \leq \zeta_*$ \label{lem_cross}
  \item $ \sqrt{1-\zeta_*^2} \leq \sigma_i((\I-\Phat \Phat{}')\bm{Q})\leq 1 $
\end{enumerate}
\end{lem}


We begin by first bounding the RIC of the CS matrix $\bm{\Phi}_{(k)}.$

\begin{lem}[{Bounding the RIC of $\bm{\Phi}_{(k)}$ \cite[Lemma 6.6]{ReProCS_IT}}] \label{RIC_bnd}
Recall that $\zeta_*:= \|(\I-\Phat_*\Phat_*{}')\bm{P}_*\|_2$.  The following hold.
\begin{enumerate}
\item Suppose that a basis matrix $\bm{P}$ can be split as $\bm{P} = [\bm{P}_1, \bm{P}_2]$ where $\bm{P}_1$ and $\bm{P}_2$ are also basis matrices. Then $\kappa_s^2 (\bm{P}) = \max_{\mathcal{T}: |\mathcal{T}| \leq s} \|{\I_T}'\bm{P}\|_2^2 \le \kappa_s^2 (\bm{P}_1) + \kappa_s^2 (\bm{P}_2)$.
\item $\kappa_s^2(\Phat_*) \leq (\kappa_{s,*})^2 + 2\zeta_*$
\item $\kappa_s (\Phat_{\rmnew,k}) \leq \kappa_{s,\rmnew} + \zeta_k + \zeta_*$
\item $\delta_{s} (\bm{\Phi}_{(0)}) = \kappa_s^2 (\Phat_*) \leq  (\kappa_{s,*})^2 + 2 \zeta_*$
\item $\delta_{s}(\bm{\Phi}_{(k)})  = \kappa_s^2 ([\Phat_* \ \Phat_{\rmnew,k}]) \leq \kappa_s^2 (\Phat_*) + \kappa_s^2 (\Phat_{\rmnew,k}) \leq (\kappa_{s,*})^2 + 2\zeta_* + (\kappa_{s,\rmnew} +  \zeta_k + \zeta_*)^2$ for $k \ge 1$
\end{enumerate}
\end{lem}

\begin{proof}
\begin{enumerate}
\item Since $\bm{P}$ is a basis matrix, $\kappa_s^2 (\bm{P}) = \max_{|\mathcal{T}| \leq s} \|{\I_\mathcal{T}}' \bm{P}\|_2^2$. Also, $\|{\I_\mathcal{T}}' \bm{P}\|_2^2 = \|{\I_\mathcal{T}}' [\bm{P}_1, \bm{P}_2] [\bm{P}_1, \bm{P}_2]' \I_\mathcal{T} \|_2 =  \|{\I_\mathcal{T}}' (\bm{P}_1 {\bm{P}_1}' + \bm{P}_2 {\bm{P}_2}') \I_\mathcal{T} \|_2 \le \|{\I_\mathcal{T}}' \bm{P}_1 {\bm{P}_1}' \I_\mathcal{T}\|_2 + \|{\I_\mathcal{T}}' \bm{P}_2 {\bm{P}_2}' \I_\mathcal{T} \|_2$. Thus, the inequality follows.

\item For any set $\mathcal{T}$ with $|\mathcal{T}| \le s$, $\|{\I_\mathcal{T}}' \Phat_*\|_2^2  = \|{\I_\mathcal{T}}' \Phat_* \Phat_*{}' \I_\mathcal{T}\|_2 =\|{\I_\mathcal{T}}'( \Phat_* \Phat_*{}' - \bm{P}_* {\bm{P}_*}' + \bm{P}_* {\bm{P}_*}') \I_\mathcal{T}\|_2 \leq \|{\I_\mathcal{T}}'( \Phat_* \Phat_*{}' - \bm{P}_* {\bm{P}_*}' ) \I_\mathcal{T} \|_2 + \|{\I_\mathcal{T}}' \bm{P}_* {\bm{P}_*}' \I_\mathcal{T}\|_2 \leq 2\zeta_* + (\kappa_{s,*})^2$. The last inequality follows using Lemma \ref{lemma0} with $\bm{P} = \bm{P}_*$ and $\hat{\bm{P}} = \hat{\bm{P}}_*$.

\item By Lemma \ref{lemma0} with $\bm{P} = \bm{P}_*$, $\Phat = \Phat_*$ and $\bm{Q} = \bm{P}_\rmnew$, $\|{\bm{P}_{\rmnew}}' \Phat_* \|_2 \leq \zeta_*$.
By Lemma \ref{lemma0} with $\bm{P} = \bm{P}_\rmnew$ and $\hat{\bm{P}} = \hat{\bm{P}}_{\rmnew,k}$, $\|( \I - \bm{P}_\rmnew {\bm{P}_\rmnew}' ) \Phat_{\rmnew,k} \|_2  = \|( \I - \Phat_{\rmnew,k}\Phat_{\rmnew,k}{}') \bm{P}_{\rmnew}\|_2$.

For any set $\mathcal{T}$ with $|\mathcal{T}| \leq s$, $\|{\I_\mathcal{T}}' \Phat_{\rmnew,k}\|_2 \leq \|{\I_\mathcal{T}} '( \I - \bm{P}_{\rmnew} {\bm{P}_{\rmnew}}') \Phat_{\rmnew,k} \|_2 + \|{\I_\mathcal{T}}' \bm{P}_{\rmnew} {\bm{P}_{\rmnew}}' \Phat_{\rmnew,k}\|_2 \leq \|( \I - \bm{P}_{\rmnew}{\bm{P}_{\rmnew}}') \Phat_{\rmnew,k} \|_2 + \|{\I_\mathcal{T}}' \bm{P}_{\rmnew}\|_2 =  \|(\I - \Phat_{\rmnew,k}\Phat_{\rmnew,k}{}') \bm{P}_{\rmnew}\|_2 + \|{\I_\mathcal{T}}' \bm{P}_{\rmnew}\|_2 \leq  \|\bm{D}_{\rmnew,k}\|_2 + \| \Phat_* \Phat_*{}' \bm{P}_{\rmnew}\|_2 + \|{\I_\mathcal{T}}' \bm{P}_{\rmnew}\|_2 $. Taking $\max$ over $|\mathcal{T}| \le s$ the claim follows.

\item This follows using Lemma \ref{kappadelta} and the second claim of this lemma.

\item This follows using Lemma \ref{kappadelta} and the first three claims of this lemma.
\end{enumerate}

\end{proof}

\begin{corollary}\label{RICnumbnd}
If the conditions of Theorem \ref{thm1} are satisfied, and $X_{j,k-1}\in \Gamma_{j,k-1}$,  then
\begin{enumerate}
\item $\delta_s(\bm{\Phi}_{(0)}) \leq \delta_{2s}(\bm{\Phi}_{(0)})  \leq (\kappa_{2s,*})^2 + 2\zeta_*^+ <0.1 < 0.1479$
\item $\delta_s(\bm{\Phi}_{(k-1)}) \leq \delta_{2s}(\bm{\Phi}_{(k-1)}) \leq (\kappa_{2s,*})^2 + 2\zeta_*^+ + (\kappa_{2s,\rmnew} +  \zeta_{k-1}^+ + \zeta_*^+)^2 < 0.1479$
\item $\|  [ ({\bm{\Phi}_{(k-1)})_{\mathcal{T}_t}}'(\bm{\Phi}_{(k-1)})_{\mathcal{T}_t}]^{-1} \|_2 \le \frac{1}{1-\delta_s(\bm{\Phi}_{(k-1)})} < 1.2 := \phi^+$
\end{enumerate}
\end{corollary}

\begin{proof}
This follows using Lemma \ref{RIC_bnd}, the definition of $\Gamma_{j,k-1}$, and the bound on $\zeta_{k-1}^+$ from Lemma \ref{expzeta}.
\end{proof}

The following are straightforward bounds that will be useful for the proof of Lemma \ref{cslem}.

\begin{fact}\label{constants}
Under the assumptions of Theorem \ref{thm1}:
\begin{itemize}
\item $ \zeta_{j,*}^+ \gamma \le \frac{\sqrt{\zeta}}{\sqrt{r_0 + (J-1)c}} \le \sqrt{\zeta}$
\item $\zeta_{k-1}^+ \leq 0.72^{k-1} + 0.83 c\zeta$ (from Lemma \ref{expzeta})
\item $\zeta_{k-1}^+ \gamma_{\rmnew} \leq 0.72^{k-1} \gamma_{\rmnew} + 0.83c\zeta \gamma_{\rmnew} \leq 0.72^{k-1}\gamma_{\rmnew} + 0.83\sqrt{\zeta} $

\end{itemize}
\end{fact}

\begin{proof}[Proof of Lemma \ref{cslem}]
Recall that $X_{j,k-1} \in \Gamma_{j,k-1}$ implies that $\zeta_{j,*} \leq \zeta_{j,*}^+$ and $\zeta_{k-1}\leq \zeta_{k-1}^+$.
\begin{enumerate}
\item
\begin{enumerate}
\item For $t \in \mathcal{I}_{j,k}$, $\bm{b}_t := ( \I - \Phat_{t-1} \Phat_{t-1} {}') \lt = \bm{D}_{*,k-1} \bm{a}_{t,*} + \bm{D}_{\rmnew,k-1} \bm{a}_{t,\rmnew} $. Thus, using Fact \ref{constants}
\begin{align*}
\|\bm{b}_t\|_2 & \leq \zeta_{j,*} \sqrt{r} \gamma + \zeta_{k-1} \sqrt{c} \gamma_{\rmnew} \\
& \leq \sqrt{\zeta}\sqrt{r} + (0.72^{k-1}\gamma_{\rmnew} + .83\sqrt{\zeta})\sqrt{c} \\
& = \sqrt{c} 0.72^{k-1} \gamma_{\rmnew} + \sqrt{\zeta} (\sqrt{r} + 0.83\sqrt{c}) \leq \xi.
\end{align*}

\item By Corollary \ref{RICnumbnd}, $\delta_{2s} (\bm{\Phi}_{(k-1)}) < 0.15< \sqrt{2}-1$. Given $|\mathcal{T}_t| \leq s$, $\|\bm{b}_t\|_2 \leq \xi$, by the theorem in \cite{candes_rip}, the CS error satisfies
\[
\|\hat{\bm{x}}_{t,\cs} - \xt \|_2 \leq  \frac{4\sqrt{1+\delta_{2s}(\bm{\Phi}_{(k-1)})}}{1-(\sqrt{2}+1)\delta_{2s}(\bm{\Phi}_{(k-1)})} \xi < 7 \xi.
\]

\item Using the above, $\|\hat{\bm{x}}_{t,\cs} - \xt\|_{\infty} \leq 7  \xi$. Since $\min_{i\in \mathcal{T}_t} |(\xt)_{i}| \geq x_{\min}$ and $(\xt)_{\mathcal{T}_t^c} = 0$, $\min_{i\in \mathcal{T}_t} |(\hat{\bm{x}}_{t,\cs})_i| \geq x_{\min} - 7 \xi$ and $\min_{i \in \mathcal{T}_t^c} |(\hat{\bm{x}}_{t,\cs})_i| \leq 7 \xi$. If $\omega < x_{\min} - 7 \xi$, then $\hat{\mathcal{T}}_t \supseteq \mathcal{T}_t$. On the other hand, if $\omega > 7 \xi$, then $\hat{\mathcal{T}}_t \subseteq \mathcal{T}_t$. Since $\omega$ satisfies $7 \xi \leq \omega \leq x_{\min} -7 \xi$, the support of $\xt$ is exactly recovered, i.e. $\hat{\mathcal{T}}_t = \mathcal{T}_t$.

\item Given $\hat{\mathcal{T}}_t = \mathcal{T}_t$, the least squares estimate of $\xt$ satisfies $(\hat{\bm{x}}_t)_{\mathcal{T}_t} = [(\bm{\Phi}_{(k-1)})_{\mathcal{T}_t}]^{\dag} \bm{y}_t = [ (\bm{\Phi}_{(k-1)})_{\mathcal{T}_t}]^{\dag} (\bm{\Phi}_{(k-1)} \xt + \bm{\Phi}_{(k-1)} \lt)$ and $(\hat{\bm{x}}_t)_{\mathcal{T}_t^c} = 0$ for $t \in \mathcal{I}_{j,k}$. 
Also,  ${(\bm{\Phi}_{(k-1)})_{\mathcal{T}_t}}' \bm{\Phi}_{(k-1)} = {\I_{\mathcal{T}_t}}' \bm{\Phi}_{(k-1)}$ (this follows since $(\bm{\Phi}_{(k-1)})_{\mathcal{T}_t} = \bm{\Phi}_{(k-1)} \I_{\mathcal{T}_t}$ and ${\bm{\Phi}_{(k-1)}}'\bm{\Phi}_{(k-1)} = \bm{\Phi}_{(k-1)}$).  
Using this, the LS error $\et := \hat{\bm{x}}_t - \xt$ satisfies (\ref{etdef0}).
Thus, using Fact \ref{constants} and the bounds on $\|\bm{a}_t\|_{\infty}$ and $\|\bm{a}_{t,\rmnew}\|_{\infty}$,
\begin{align*}
\|\bm{e}_t\|_2 &\leq
\begin{cases}
 \phi^+ (\zeta_{j,*}^+ \sqrt{r}\gamma +  \zeta_{j,k-1}^+ \sqrt{c}\gamma_{\rmnew}) & t \in [t_j,t_j +d]\\
\phi^+ (\zeta_{j,*}^+ \sqrt{r}\gamma +  c\zeta\sqrt{c}\gamma) & t \in (t_j +d,t_{j+1})
\end{cases}\\
 &\le
\begin{cases}
1.2 \left(1.83\sqrt{\zeta} +  (0.72)^{k-1}\sqrt{c}\gamma_{\rmnew}  \right) &  t \in [t_j,t_j +d] \\
1.2\left( 2\sqrt{\zeta} \right) & t \in (t_j + d, t_{j+1})
\end{cases}
\end{align*}

\end{enumerate}

\item The second claim is just a restatement of the first.

\end{enumerate}

\end{proof}

\section{Cauchy-Schwarz inequality}

\begin{lem}[Cauchy-Schwarz for a sum of vectors]\label{CSsum}
For vectors $\bm{x}_t$ and $\bm{y}_t$,
\[
\left(\sum_{t=1}^{\alpha} {\bm{x}_t}'\bm{y}_t\right) \leq \left( \sum_t \|\bm{x}_t\|_2^2 \right) \left( \sum_t \|\bm{y}_t\|_2^2 \right)
\]
\end{lem}

\begin{proof}
\begin{align*}
\left(\sum_{t=1}^{\alpha} {\bm{x}_t}'\bm{y}_t\right)^2 = \left( [ {\bm{x}_1}', \dots, {\bm{x}_{\alpha}}'] \left[\begin{array}{c}
\bm{y}_{1}\\ \vdots \\ \bm{y}_{\alpha} 
\end{array} \right] \right)^2 \leq \left\|\left[\begin{array}{c}
\bm{x}_{1}\\ \vdots \\ \bm{x}_{\alpha} 
\end{array} \right] \right\|_2^2 \left\| \left[\begin{array}{c}
\bm{y}_{1}\\ \vdots \\ \bm{y}_{\alpha} 
\end{array} \right] \right\|_2^2 = \left(\sum_{t=1}^{\alpha}\|\bm{x}_t\|_2^2\right)\left(\sum_{t=1}^{\alpha}\|\bm{y}_t\|_2^2\right)
\end{align*}
The inequality is by Cauchy-Schwarz for a single vector.
\end{proof}

\begin{lem}[Cauchy-Schwarz for a sum of matrices]\label{CSmat}
For matrices $\bm{X}_t$ and $\bm{Y}_t$,
\[
\left\|\frac{1}{\alpha} \sum_{t=1}^{\alpha} \bm{X}_t {\bm{Y}_t}'\right\|_2^2 \leq \lambda_{\max}\left(\frac{1}{\alpha} \sum_{t=1}^{\alpha} \bm{X}_t {\bm{X}_t}'\right)
\lambda_{\max}\left(\frac{1}{\alpha} \sum_{t=1}^{\alpha} \bm{Y}_t {\bm{Y}_t}'\right)
\]
\end{lem}

\begin{proof}
\begin{align*}
\left\| \sum_{t=1}^{\alpha} \bm{X}_t {\bm{Y}_t}'\right\|_2^2 &=
\max_{\substack{\|\bm{x}\|=1\\\|\bm{y}\|=1}} \left| \bm{x}'\left(\sum_{t}\bm{X}_t {\bm{Y}_t}'\right)\bm{y} \right|^2 \\
&= \max_{\substack{\|\bm{x}\|=1\\\|\bm{y}\|=1}} \left| \sum_{t=1}^{\alpha} ({\bm{X}_t}'\bm{x})'({\bm{Y}_t}'\bm{y})  \right|^2 \\
&\leq \max_{\substack{\|\bm{x}\|=1\\\|\bm{y}\|=1}} \left( \sum_{t=1}^{\alpha} \left\| {\bm{X}_t}'\bm{x} \right\|_2^2\right) \left( \sum_{t=1}^{\alpha} \left\| {\bm{Y}_t}'\bm{y} \right\|_2^2\right) \\
&= \max_{\|\bm{x}\|=1}  \bm{x}' \sum_{t=1}^{\alpha} \bm{X}_t{\bm{X}_t}' \ \bm{x}  \ \cdot \ \max_{\|\bm{y}\|=1} \bm{y}' \sum_{t=1}^{\alpha} \bm{Y}_t{\bm{Y}_t}' \ \bm{y}\\
&= \lambda_{\max}\left( \sum_{t=1}^{\alpha} \bm{X}_t{\bm{X}_t}' \right)\lambda_{\max}\left( \sum_{t=1}^{\alpha} \bm{Y}_t{\bm{Y}_t}' \right)
\end{align*}
The inequality is by Lemma \ref{CSsum}. The penultimate line is because $\|\bm{x}\|_2^2 = {\bm{x}'\bm{x}}$.
Multiplying both sides by $\left(\frac{1}{\alpha}\right)^2$ gives the desired result.
\end{proof}

\bibliographystyle{IEEEtran}
\bibliography{tipnewpfmtNIPS}

\begin{thebibliography}{10}
\providecommand{\url}[1]{#1}
\csname url@samestyle\endcsname
\providecommand{\newblock}{\relax}
\providecommand{\bibinfo}[2]{#2}
\providecommand{\BIBentrySTDinterwordspacing}{\spaceskip=0pt\relax}
\providecommand{\BIBentryALTinterwordstretchfactor}{4}
\providecommand{\BIBentryALTinterwordspacing}{\spaceskip=\fontdimen2\font plus
\BIBentryALTinterwordstretchfactor\fontdimen3\font minus
  \fontdimen4\font\relax}
\providecommand{\BIBforeignlanguage}[2]{{%
\expandafter\ifx\csname l@#1\endcsname\relax
\typeout{** WARNING: IEEEtran.bst: No hyphenation pattern has been}%
\typeout{** loaded for the language `#1'. Using the pattern for}%
\typeout{** the default language instead.}%
\else
\language=\csname l@#1\endcsname
\fi
#2}}
\providecommand{\BIBdecl}{\relax}
\BIBdecl

\bibitem{error_correction_PCP_l1}
J.~Wright and Y.~Ma, ``Dense error correction via l1-minimization,'' \emph{IEEE
  Trans. on Info. Th.}, vol.~56, no.~7, pp. 3540--3560, 2010.

\bibitem{rpca}
E.~J. Cand{\`e}s, X.~Li, Y.~Ma, and J.~Wright, ``Robust principal component
  analysis?'' \emph{Journal of ACM}, vol.~58, no.~3, 2011.

\bibitem{rpca2}
V.~Chandrasekaran, S.~Sanghavi, P.~A. Parrilo, and A.~S. Willsky,
  ``Rank-sparsity incoherence for matrix decomposition,'' \emph{SIAM Journal on
  Optimization}, vol.~21, 2011.

\bibitem{chen2011low}
Y.~Chen, A.~Jalali, S.~Sanghavi, and C.~Caramanis, ``Low-rank matrix recovery
  from errors and erasures,'' in \emph{Information Theory Proceedings (ISIT),
  2011 IEEE International Symposium on}.\hskip 1em plus 0.5em minus 0.4em\relax
  IEEE, 2011, pp. 2313--2317.

\bibitem{outlier_pursuit}
H.~Xu, C.~Caramanis, and S.~Sanghavi, ``Robust pca via outlier pursuit,''
  \emph{IEEE Tran. on Information Theorey}, vol.~58, no.~5, 2012.

\bibitem{rpca_tropp}
M.~B. McCoy and J.~A. Tropp, ``Sharp recovery bounds for convex deconvolution,
  with applications,'' arXiv:1205.1580.

\bibitem{linear_inverse_prob}
V.~Chandrasekaran, B.~Recht, P.~A. Parrilo, and A.~S. Willsky, ``The convex
  geometry of linear inverse problems,'' \emph{Foundations of Computational
  Mathematics}, vol.~12, no.~6, 2012.

\bibitem{noisy_undersampled_yuan}
M.~Tao and X.~Yuan, ``Recovering low-rank and sparse components of matrices
  from incomplete and noisy observations,'' \emph{SIAM Journal on
  Optimization}, vol.~21, no.~1, pp. 57--81, 2011.

\bibitem{Torre03aframework}
F.~D.~L. Torre and M.~J. Black, ``A framework for robust subspace learning,''
  \emph{International Journal of Computer Vision}, vol.~54, pp. 117--142, 2003.

\bibitem{mardani2013dynamic}
M.~Mardani, G.~Mateos, and G.~B. Giannakis, ``Dynamic anomalography: Tracking
  network anomalies via sparsity and low rank,'' \emph{Selected Topics in
  Signal Processing, IEEE Journal of}, vol.~7, no.~1, pp. 50--66, 2013.

\bibitem{ReProCS_IT}
C.~Qiu, N.~Vaswani, B.~Lois, and L.~Hogben, ``Recursive robust pca or recursive
  sparse recovery in large but structured noise,'' \emph{arXiv: 1211.3754
  [cs.IT]}, 2014.

\bibitem{han_tsp}
H.~Guo, C.~Qiu, and N.~Vaswani, ``An online algorithm for separating sparse and
  low-dimensional signal sequences from their sum,'' \emph{arXiv:1310.4261
  [cs.IT]}.

\bibitem{GRASTA}
J.~He, L.~Balzano, and J.~Lui, ``Online robust subspace tracking from partial
  information,'' \emph{arXiv:1109.3827 [cs.IT]}.

\bibitem{mateos2012robust}
G.~Mateos and G.~B. Giannakis, ``Robust pca as bilinear decomposition with
  outlier-sparsity regularization,'' \emph{Signal Processing, IEEE Transactions
  on}, vol.~60, no.~10, pp. 5176--5190, 2012.

\bibitem{zhan_correlated}
J.~Zhan and N.~Vaswani, ``Performance guarantees for reprocs -- correlated
  low-rank matrix entries case,'' arXiv:1405.5887 [cs.IT].

\bibitem{rrpcp_globalsip}
B.~Lois, N.~Vaswani, and C.~Qiu, ``Performance guarantees for undersampled
  recursive sparse recovery in large but structured noise,'' in
  \emph{GlobalSIP}, 2013.

\bibitem{OnlinePCA_ContaminatedData}
\BIBentryALTinterwordspacing
J.~Feng, H.~Xu, S.~Mannor, and S.~Yan, ``Online pca for contaminated data,'' in
  \emph{Advances in Neural Information Processing Systems 26}, C.~Burges,
  L.~Bottou, M.~Welling, Z.~Ghahramani, and K.~Weinberger, Eds.\hskip 1em plus
  0.5em minus 0.4em\relax Curran Associates, Inc., 2013, pp. 764--772.
  [Online]. Available:
  \url{http://papers.nips.cc/paper/5135-online-pca-for-contaminated-data.pdf}
\BIBentrySTDinterwordspacing

\bibitem{rpca_stochatistic_optimization}
\BIBentryALTinterwordspacing
J.~Feng, H.~Xu, and S.~Yan, ``Online robust pca via stochastic optimization,''
  in \emph{Advances in Neural Information Processing Systems 26}, C.~Burges,
  L.~Bottou, M.~Welling, Z.~Ghahramani, and K.~Weinberger, Eds.\hskip 1em plus
  0.5em minus 0.4em\relax Curran Associates, Inc., 2013, pp. 404--412.
  [Online]. Available:
  \url{http://papers.nips.cc/paper/5131-online-robust-pca-via-stochastic-optimization.pdf}
\BIBentrySTDinterwordspacing

\bibitem{mod_pcp}
J.~Zhan and N.~Vaswani, ``Robust pca with partial subspace knowledge,''
  \emph{arXiv:1403.1591 [cs.IT]}.

\bibitem{hsu2011robust}
D.~Hsu, S.~M. Kakade, and T.~Zhang, ``Robust matrix decomposition with sparse
  corruptions,'' \emph{Information Theory, IEEE Transactions on}, vol.~57,
  no.~11, pp. 7221--7234, 2011.

\bibitem{candes_rip}
E.~Cand{\`e}s, ``The restricted isometry property and its implications for
  compressed sensing,'' \emph{Compte Rendus de l'Academie des Sciences, Paris,
  Serie I}, pp. 589--592, 2008.

\bibitem{nadler}
B.~Nadler, ``Finite sample approximation results for principal component
  analysis: A matrix perturbation approach,'' \emph{The Annals of Statistics},
  vol.~36, no.~6, 2008.

\bibitem{davis_kahan}
C.~Davis and W.~M. Kahan, ``The rotation of eigenvectors by a perturbation.
  iii,'' \emph{SIAM Journal on Numerical Analysis}, Mar. 1970.

\bibitem{tail_bound}
J.~A. Tropp, ``User-friendly tail bounds for sums of random matrices,''
  \emph{Foundations of Computational Mathematics}, vol.~12, no.~4, 2012.

\bibitem{hornjohnson}
R.~Horn and C.~Johnson, \emph{Matrix Analysis}.\hskip 1em plus 0.5em minus
  0.4em\relax Cambridge University Press, 1985.

\bibitem{longest_run}
M.~Muselli, ``On convergence properties of pocket algorithm,'' \emph{Neural
  Networks, IEEE Transactions on}, vol.~8, no.~3, pp. 623--629, May 1997.

\end{thebibliography}

\end{document}